\documentclass[a4paper,UKenglish,numberwithinsect]{lipics-v2018-arxiv}
 
\usepackage{anyfontsize}
\usepackage{microtype}
\usepackage[absolute]{textpos}
\usepackage{enumerate}
\usepackage[shortlabels,inline]{enumitem}

\usepackage{relsize,xspace}
\usepackage{xcolor}
\usepackage{mathtools}
\usepackage{todonotes}
\usepackage{comment}
\usepackage{microtype}
\usepackage{amsmath}
\usepackage{amssymb}
\usepackage{amsfonts}
\usepackage{mathrsfs}
\usepackage{bm}
\usepackage{tikz}
\usepackage{refcount}
\usepackage{wrapfig}
\usepackage{cite}
\usepackage{float}

\newcommand{\Oof}{\mathcal{O}}

\newcommand{\Pp}{\mathcal{P}}
\newcommand{\Qq}{\mathcal{Q}}

\newcommand{\yes}{{\sf  yes}}

\newtheorem{observation}[theorem]{Observation}
\newcounter{claimcounter}
\newtheorem{claim}[claimcounter]{Claim}

\usepackage{ragged2e}

\newcommand{\defparprob}[4]{
  \vspace{1mm}
\noindent\fbox{
  \begin{minipage}{0.96\textwidth}
  \begin{tabular*}{\textwidth}{@{\extracolsep{\fill}}lr} #1  & {\bf{Parameter:}} #3 \\ \end{tabular*}
  {\bf{Input:}} #2  \\
  {\bf{Question:}} #4
  \end{minipage}
  }
  \vspace{1mm}
}

\tikzstyle{vertex}=[circle,inner sep=2,minimum size =2mm,semithick,fill=white!80!blue, draw=black]

\nolinenumbers
\title{On the Parameterized Complexity of Reconfiguration of Connected Dominating Sets}

\titlerunning{Reconfiguration of Connected Dominating Set}

\author{Daniel Lokshtanov}{University of California Santa Barbara, Santa Barbara, USA}{daniello@ucsb.edu}{}{}
\author{Amer E. Mouawad}{Department of Computer Science, American University of Beirut, Lebanon}{aa368@aub.edu.lb}{}{}
\author{Fahad Panolan}{Department of Computer Science and Engineering, IIT Hyderabad, India}{fahad@iith.ac.in}{}{}
\author{Sebastian Siebertz}{University of Bremen, Germany}{siebertz@uni-bremen.de}{}{}
\authorrunning{D.\ Lokshtanov et al.} 

\Copyright{Daniel Lokshtanov, Amer E. Mouawad, Fahad Panolan, Sebastian Siebertz}

\subjclass{CCS $\rightarrow$  Theory of computation $\rightarrow$ Design and analysis of algorithms $\rightarrow$ Parameterized complexity and exact algorithms}


\category{}

\relatedversion{}

\supplement{}

\funding{}

\acknowledgements{}

\EventEditors{John Q. Open and Joan R. Access}
\EventNoEds{2}
\EventLongTitle{}
\EventShortTitle{}
\EventAcronym{}
\EventYear{2020}
\EventDate{}
\EventLocation{}
\EventLogo{}
\SeriesVolume{}
\ArticleNo{}

\begin{document}
\maketitle

\begin{abstract}
  In a reconfiguration version of an optimization problem
  $\mathcal{Q}$ the input is an instance of $\Qq$ and two feasible
  solutions $S$ and $T$. The objective is to determine whether there exists 
  a step-by-step transformation between $S$ and $T$ such that all
  intermediate steps also constitute feasible solutions. In this work, 
  we study the parameterized complexity of the \textsc{Connected
    Dominating Set Reconfiguration} problem (\textsc{CDS-R)}. It was
  shown in previous work that the \textsc{Dominating Set
    Reconfiguration} problem (\textsc{DS-R}) parameterized by $k$, the
  maximum allowed size of a dominating set in a reconfiguration
  sequence, is fixed-parameter tractable on all graphs that  exclude a biclique $K_{d,d}$ as a subgraph, for some
  constant $d \geq 1$. We show that the
  additional connectivity constraint makes the problem much harder,
  namely, that \textsc{CDS-R} is \textsf{W}$[1]$-hard parameterized by $k+\ell$, the maximum allowed size of a dominating set plus the 
  length of the reconfiguration sequence, already on
  \mbox{$5$-degenerate} graphs.  On the positive side, we show that
  \textsc{CDS-R} parameterized by $k$ is fixed-parameter tractable, and in fact admits a
  polynomial kernel on planar graphs.
\end{abstract}

\pagenumbering{arabic}

\section{Introduction}

In an optimization problem $\Qq$, we are usually asked to determine the
existence of a feasible solution for an instance $\mathcal{I}$ of $\Qq$. In a \emph{reconfiguration version} of $\Qq$, we
are instead given a source feasible solution $S$ and a target feasible
solution $T$ and we are asked to determine whether it is possible to
transform $S$ into $T$ by a sequence of step-by-step transformations
such that after each intermediate step we also maintain feasible solutions. 
Formally, we consider a graph, called the \emph{reconfiguration
  graph}, that has one vertex for each feasible solution and where two
vertices are connected by an edge if we allow the transformation
between the two corresponding solutions. We are then asked to
determine whether $S$ and $T$ are connected in the reconfiguration
graph, or even to compute a shortest path between them.
Historically, the study of reconfiguration questions predates the
field of computer science, as many classic one-player games can be
formulated as such reachability questions~\cite{JS79,KPS08}, e.g., the
$15$-puzzle and Rubik's cube.  More recently, reconfiguration problems
have emerged from computational problems in different areas such as
graph theory~\cite{CHJ08,IDHPSUU11,IKD12}, constraint
satisfaction~\cite{GKMP09,DBLP:conf/icalp/MouawadNPR15} and computational
geometry~\cite{DBLP:books/daglib/0019278,DBLP:conf/stacs/KanjX15,DBLP:journals/comgeo/LubiwP15},
and even quantum complexity
theory~\cite{DBLP:conf/icalp/GharibianS15}.  Reconfiguration problems
have been receiving considerable attention in recent literature, we
refer the reader to~\cite{nishimura2018introduction,
  van2013complexity,mouawad2015reconfiguration} for an extensive
overview.

In this work, we consider the \textsc{Connected Dominating Set
  Reconfiguration} problem (\textsc{CDS-R}) in undirected graphs. A
\emph{dominating set} in a graph $G$ is a set $D\subseteq V(G)$ such
that every vertex of $G$ lies either in $D$ or is adjacent to a vertex
in $D$. A dominating set~$D$ is a \emph{connected dominating set} if
the graph induced by $D$ is connected. The
\textsc{Dominating Set} problem and its connected variant have many
applications, including the modeling of facility location problems,
routing problems, and many more.

We study \textsc{CDS-R} under the \emph{Token Addition/Removal} model
(\textsf{TAR} model).  Suppose we are given a connected dominating
set $D$ of a graph $G$, and imagine that a token/pebble is placed on
each vertex in $D$. The \textsf{TAR} rule allows either the addition or
removal of a single token/pebble at a time from $D$, if this results in a
connected dominating set of size at most a given bound $k\geq 1$.
%
%
%
A sequence $D_1, \ldots, D_\ell$ of connected dominating sets of a
graph $G$ is called a \emph{reconfiguration sequence} between $D_1$
and $D_\ell$ under \textsc{TAR} if the change from~$D_i$ to $D_{i+1}$
respects the \textsf{TAR} rule, for $1\leq i< \ell$.  The
\emph{length} of the reconfiguration sequence is~$\ell-1$.
%
The \textsc{(Connected) Dominating Set Reconfiguration} problem for
\textsf{TAR} gets as input a graph $G$, two (connected) dominating
sets $S$ and $T$ and an integer $k\geq 1$, and the task is to decide
whether there exists a reconfiguration sequence between $S$ and $T$
under \textsf{TAR} using at most $k$ tokens/pebbles.


Structural properties of the reconfiguration graph for $k$-dominating
sets were studied in~\cite{haas2014k,suzuki2016reconfiguration}. The \textsc{Dominating Set Reconfiguration} problem was shown
to be \textsc{PSPACE}-complete in~\cite{mouawad2017parameterized},
even on split graphs, bipartite graphs, planar graphs and graphs of
bounded bandwidth.  Both pathwidth and treewidth of a graph are
bounded by its bandwidth, hence the \textsc{Dominating Set
  Reconfiguration} problem is \textsc{PSPACE}-complete on graphs of
bounded pathwidth and treewidth.  These hardness results motivated the
study of the parameterized complexity of the problem. It was shown
in~\cite{mouawad2017parameterized} that the \textsc{Dominating Set
  Reconfiguration} problem is $\textsf{W}[2]$-hard when parameterized
by $k+\ell$, where $k$ is the bound on the number of tokens and $\ell$
is the length of the reconfiguration sequence. However, the problem
becomes fixed-parameter tractable on graphs that exclude a fixed
complete bipartite graph $K_{d,d}$ as a subgraph, as shown
in~\cite{LokshtanovMPRS18}. Such so-called \emph{biclique-free}
classes are very general sparse graph classes, including in particular
the planar graphs, which are $K_{3,3}$-free.

In this work we study the complexity of \textsc{CDS-R}. The standard
reduction from \textsc{Dominating Set} to \textsc{Connected Dominating
  Set} shows that also \textsc{CDS-R} is \textsc{PSPACE}-complete,
even on graphs of bounded pathwidth (Figure~\ref{fig-ds-red}).

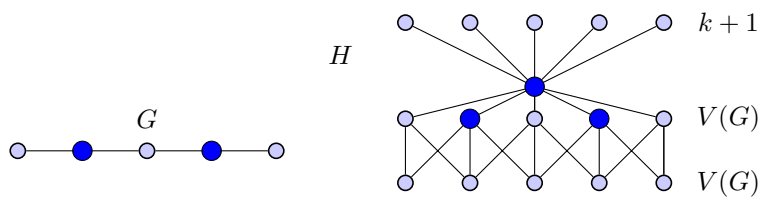
\begin{figure}[ht]
\begin{center}
\begin{tikzpicture}[scale=0.85]
\begin{scope}[yshift=0.5cm]
\node at (2,0.5) {$G$};
\node[vertex] (m0) at (0,0) {}; 
\node[vertex] (m1) at (1,0) {};
\node[vertex] (m2) at (2,0) {};
\node[vertex] (m3) at (3,0) {};
\node[vertex] (m4) at (4,0) {};
\draw[fill=blue] (1,0) circle (1.5mm);
\draw[fill=blue] (3,0) circle (1.5mm);
\draw[-] (m0) -- (m1) -- (m2) -- (m3) -- (m4);
\end{scope}

\begin{scope}[xshift=5cm]
\node[vertex] (n0) at (3,1.5){};

\foreach \i in {1,...,5}{
  
  \node[vertex] (a\i) at (\i,1){};
  \node[vertex] (b\i) at (\i,0){};
  \node[vertex] (n\i) at (\i,2.5){};
}
\foreach \i in {1,...,5}{
  \draw (n0) -- (n\i);
  \draw (n0) -- (a\i);
  \draw (a\i) -- (b\i);
}
\draw[-] (a1) -- (b2) -- (a3) -- (b4) -- (a5) -- (b5) -- (a4) -- (b3) -- (a2) -- (b1);
\draw (n0) -- (a3);
\draw[fill=blue] (2,1) circle (1.5mm);
\draw[fill=blue] (4,1) circle (1.5mm);
\draw[fill=blue] (3,1.5) circle (1.5mm);

\node at (-0.2,0) {\textcolor{white}{$G$}};
\node at (6,0) {$V(G)$};
\node at (6,1) {$V(G)$};
\node at (6,2.5) {$k+1$};
\node at (0,2) {$H$};
\end{scope}
\end{tikzpicture}
\caption{A graph $G$ with a minimum dominating set 
of size $k=2$ marked in dark blue and the graph $H$ 
obtained in the standard reduction from \textsc{Dominating 
Set} to \textsc{Connected 
Dominating Set}. $G$ has a dominating set of size $k$ if and
only if $H$ has a connected dominating set of size $k+1$. 
If $p$ is equal to the pathwidth of $G$ then the 
pathwidth of $H$ is bounded by $2p+1$.}
\end{center}
\end{figure}\label{fig-ds-red}

We hence turn our attention to
the parameterized complexity of the problem.  We first show that the
additional connectivity constraint makes the problem much harder,
namely, that \textsc{CDS-R} parameterized by $k+\ell$ is
\textsf{W}$[1]$-hard already on $5$-degenerate graphs. As
$5$-degenerate graphs exclude the biclique $K_{6,6}$ as a subgraph,
\textsc{Dominating Set Reconfiguration} is fixed-parameter tractable
on much more general graph classes than its connected variant. 
To prove hardness we first introduce an auxiliary problem that
we believe is of independent interest. In the \textsc{Colored
Connected Subgraph} problem we are given a graph $G$, an integer~$k$, and a 
coloring $c\colon V(G)\rightarrow C$, for some color set $C$
with $|C|\leq k$. The question is whether $G$ contains a 
vertex subset $H$ on at most $k$ vertices such that $G[H]$ is connected and $H$ contains at least 
one vertex of every color in $C$ (i.e., $c(V(H))=C$). The reconfiguration 
variant \textsc{Colored Connected Subgraph Reconfiguration (CCS-R)} is defined
as expected. We first prove that \textsc{CCS-R} reduces to 
\textsc{CDS-R} by a parameter preserving reduction (where $k+\ell$ is the parameter)  
and the degeneracy of the reduced graph is at most the degeneracy of the input graph plus one. We then 
prove that the known $\textsf{W}[1]$-hard problem 
\textsc{Multicolored Clique} reduces to
\textsc{CCS-R} on $4$-degenerate graphs. The last reduction 
has the additional property that for an input $(G,c,k)$ of 
\textsc{Multicolored Clique} the resulting instance of 
\textsc{CCS-R} admits either a reconfiguration sequence
of length~$\Oof(k^3)$, or no reconfiguration sequence at
all. Hence, we derive that both \textsc{CDS-R} and 
\textsc{CCS-R} are $\textsf{W}[1]$-hard parameterized
by $k+\ell$ on $5$-degenerate and $4$-degenerate graphs, respectively. 

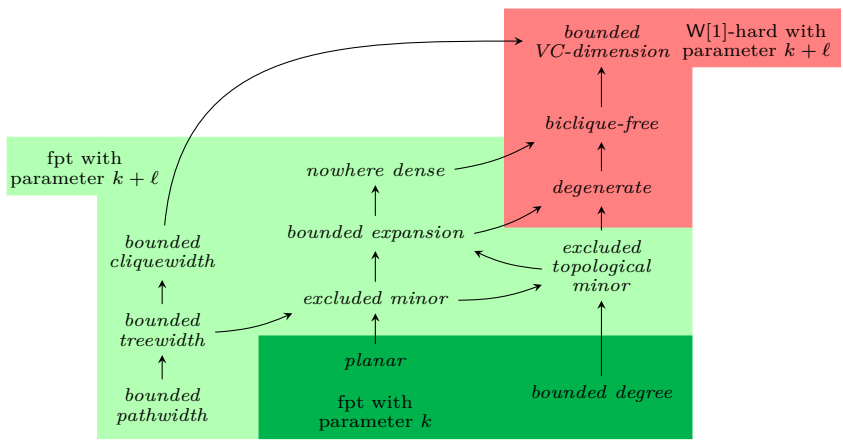
\begin{figure}[!ht]
\begin{center}
  \begin{center}
\begin{tikzpicture}[scale=0.85]
\filldraw[fill=green!30!white, draw=green!30!white] (2.2,-4.7) rectangle (8.5,0);
\filldraw[fill=green!30!white, draw=green!30!white] (8.5,-4.7) rectangle (11.4,-1.4);
\filldraw[fill=green!30!white, draw=green!30!white] (2.2,-0.9) rectangle (0.8,0);
\node[align=center] at (2,-0.5) {\scriptsize fpt with\\[-2mm] \scriptsize parameter $k+\ell$};

\filldraw[fill=red!50!white, draw=red!50!white] (8.5,-1.4) rectangle (11.4,2.0);
\filldraw[fill=red!50!white, draw=red!50!white] (11.4,2.0) rectangle (13.7,1.1);
\node[align=center] at (12.4,1.5) {\scriptsize \textsf{W}$[1]$-hard with\\[-2mm] \scriptsize parameter $k+\ell$};

\filldraw[fill=green!70!blue, draw=green!70!blue] (4.7,-4.7) rectangle (11.4,-3.1);
\node[align=center] at (6.5,-4.3) {\scriptsize fpt with\\[-2mm] \scriptsize parameter $k$};

\node (nd) at (6.5,-0.5) {\scriptsize\textit{nowhere dense}};
\node (biclique) at (10,0.2) {\scriptsize\textit{biclique-free}};
\node[align=center] (bd-VC) at (10,1.5) {\scriptsize\textit{bounded}\\[-2mm]
\scriptsize\textit{VC-dimension}};

\node (bd-deg) at (10,-4) {\scriptsize\textit{bounded degree}};
\node[align=center] (topminor) at (10,-2) {\scriptsize\textit{excluded}\\[-2mm]\scriptsize\textit{topological}\\[-2mm] \scriptsize\textit{minor}};
\node (bd-exp) at (6.5,-1.5) {\scriptsize\textit{bounded expansion}};
\node (degenerate) at (10,-0.8) {\scriptsize\textit{degenerate}};
\node (planar) at (6.5,-3.5) {\scriptsize\textit{planar}};
\node (minor) at (6.5,-2.5) {\scriptsize\textit{excluded minor}};
\node[align=center] (bd-td) at (3.2,-4.2) 
{\scriptsize\textit{bounded}\\[-2mm]\scriptsize\textit{pathwidth}};
\node[align=center] (bd-tw) at (3.2,-3) {\scriptsize\textit{bounded} \\[-2mm] \scriptsize\textit{treewidth}};
\node[align=center] (bd-cw) at (3.2,-1.8) {\scriptsize\textit{bounded} \\[-2mm] \scriptsize\textit{cliquewidth}};


\draw[->,>=stealth] (bd-td) to (bd-tw);
\draw[->,>=stealth] (bd-tw) to (bd-cw);

\draw[->,>=stealth] (planar) to (minor);
\draw[->,>=stealth] (bd-tw) to[bend right=10] (minor.south west);
\draw[->,>=stealth] (minor) to (bd-exp);
\draw[->,>=stealth] (bd-exp) to (nd);

\draw[->,>=stealth] (bd-deg) to (topminor);
\draw[->,>=stealth] (minor) to[bend right=10] (topminor);
\draw[->,>=stealth] (topminor) to[bend left=10] (bd-exp.south east);
\draw[->,>=stealth] (bd-exp.east) to[bend right=10] (degenerate.south west);
\draw[->,>=stealth] (topminor) to (degenerate);
\draw[->,>=stealth] (nd.east) to[bend right=10] (biclique.south west);
\draw[->,>=stealth] (degenerate) to (biclique);
\draw[->,>=stealth] (biclique) to (bd-VC);
\draw[->,>=stealth] (bd-cw)  .. controls (3.5,1.5) and (6, 1.5) ..  (bd-VC.west);

\end{tikzpicture}
\end{center}
\end{center}
\caption{The map of tractability for \textsc{Connected Dominating Set Reconfiguration}. 
The classes colored in dark green admit
an fpt algorithm with parameter $k$, the classes colored in light
green admit an FPT algorithm with parameter $k+\ell$. On the classes colored in red the problem is $\mathsf{W}[1]$-hard with
respect to the parameter $k+\ell$.
}
\end{figure}\label{fig-complex}

The existence of a reconfiguration sequence of length
at most $\ell$ with connected
dominating sets of size at most $k$ can be expressed by a first-order
formula of length depending only on~$k$ and $\ell$. It follows from~\cite{grohe2017deciding} that the
problem is fixed-parameter tractable parameterized by $k+\ell$ on
every nowhere dense graph class and the same is implied
by~\cite{courcelle2000linear} for every class of bounded cliquewidth.
Nowhere dense graph classes are very general classes of uniformly
sparse graphs, in particular the class of planar graphs is nowhere
dense. Nowhere dense classes are themselves biclique-free, but are not
necessarily degenerate.  Hence, our hardness result on degenerate
graphs essentially settles the question of fixed-parameter
tractability for the parameter $k+\ell$ on sparse graph classes.  It
remains an interesting open problem to find dense graph classes beyond
classes of bounded cliquewidth on which the problem is fixed-parameter
tractable.

We then turn our attention to the smaller parameter $k$ alone. We show
that \textsc{CDS-R} parameterized by $k$ is fixed-parameter tractable
on the class of planar graphs.  Our approach is as follows. We first
compute a small \emph{domination core} for $G$, a set of vertices that
captures exactly the domination properties of $G$ for dominating sets
of sizes not larger than $k$. The notion of a domination core was
introduced in the study of the \textsc{Distance-$r$ Dominating Set}
problem on nowhere dense graph classes~\cite{DawarK09}.  While the
classification of interactions with the domination core would suffice
to solve \textsc{Dominating Set Reconfiguration} on nowhere dense
classes, additional difficulties arise for the connected variant. In a
second step we use planarity to identify large subgraphs that have
very simple interactions with the domination core and prove that they
can be replaced by constant size gadgets such that the reconfiguration
properties of $G$ are preserved. 

Observe that \textsc{CDS-R} parameterized by $k$
is trivially fixed-parameter tractable on every class of bounded degree. 
The existence of a connected dominating set of size $k$ implies
that the diameter of $G$ is bounded by $k+2$, which in every
bounded degree class implies a bound on the size
of the graph depending only on the degree and $k$. 
We conjecture that \textsc{CDS-R} is
fixed-parameter tractable parameterized by $k$ on every nowhere dense
graph class. However, resolving this conjecture remains open for
future work (see Figure~\ref{fig-complex}).

The rest of the paper is organized as follows. We give background on
graph theory and fix our notation in Section~\ref{sec:prelims}.  We
show hardness of \textsc{CDS-R}  on
degenerate graphs in Section~\ref{sec:hardness} and finally show how
to handle the planar case in Section~\ref{sec:planar}. 
Due to space constraints proofs of results 
marked with a $\star$ are deferred to the appendix. 

\section{Preliminaries}\label{sec:prelims}
We denote the set of natural numbers by $\mathbb{N}$. For $n \in \mathbb{N}$, we let $[n] = \{1, 2, \dots, n\}$.  We assume that each graph $G$ is finite, simple, and undirected. We let $V(G)$ and $E(G)$ denote the vertex set and edge set of $G$, respectively. An edge between two vertices $u$ and $v$ in a graph is denoted by $\{u,v\}$ or $uv$.
 The {\em open neighborhood} of a vertex $v$ is denoted by $N_G(v) = \{u \mid \{u,v\} \in E(G)\}$ and the {\em closed neighborhood} by $N_G[v] = N_G(v) \cup \{v\}$. The degree of a vertex $v$, denoted $d_G(v)$, is $|N_G(v)|$.  For a set of vertices $S \subseteq V(G)$, we define $N_G(S) = \{v \not\in S \mid \{u,v\} \in E(G), u \in S\}$ and $N_G[S] = N_G(S) \cup S$.  The subgraph of $G$ induced by $S$ is denoted by $G[S]$, where $G[S]$ has vertex set $S$ and edge set $\{\{u,v\} \in E(G) \mid u,v \in S\}$. We let $G - S = G[V(G) \setminus S]$. A graph~$G$ is $d$-degenerate if every subgraph $H\subseteq G$ has a vertex of degree at most $d$.
For a set $C$, we use $K[C]$ to denote the complete graph on vertex set $C$. For an integer $r\in {\mathbb N}$, an $r$-independent set in a graph $G$ is a subset $U\subseteq V(G)$ such that for any two distinct vertices $u,v\in U$, the distance between~$u$ and $v$ in $G$ is more than $r$. An independent set in a graph is a $1$-independent set. 
A subset of vertices $U$ in $G$ is called a  separator in  $G$  if $G-U$ is has more than one connected component.
For $s,t\in V(G)$, we say that a subset of vertices $U$ in $G$ is an  $(s,t)$-separator in~$G$  if 
there is no path from $s$ to $t$ in $G-U$.


\section{Hardness on degenerate graphs}\label{sec:hardness}
In this section we prove that \textsc{CDS-R} and \textsc{CCS-R} are \textsc{W}$[1]$-hard when parameterized 
by~$k+\ell$ even on $5$-degenerate and $4$-degenerate graphs, respectively.  
Towards that, we first give a polynomial-time reduction from the $\textsf{W}[1]$-hard  \textsc{Multicolored Clique} problem
to \textsc{CCS-R} on $4$-degenerate graphs with the property that for an input $(G,c,k)$ of 
\textsc{Multicolored Clique} the resulting instance of \textsc{CCS-R} admits either a reconfiguration sequence
of length~$\Oof(k^3)$ or no reconfiguration sequence at all. 
As a result, we conclude that \textsc{CCS-R} is \textsc{W}$[1]$-hard when parameterized by $k+\ell$ on $4$-degenerate graphs. 
Then, we give a parameter-preserving polynomial-time reduction from \textsc{CCS-R} to \textsc{CDS-R}. 
Let us first formally define the \textsc{CCS} problem. 


\defparprob{{\sc Colored Connected Subgraph (CCS)}}
{A graph $G$, $k\in \mathbb{N}$, and a vertex-coloring $c\colon V(G)\rightarrow C$, where $|C|\leq k$}{$k$}
{Is there a vertex subset $S\subseteq V(G)$ of at most $k$ vertices 
with at least one vertex from every color class such that $G[S]$ is connected?}

\subparagraph*{Reduction from Multicolored Clique to CCS-R.}
We now present the reduction from \textsc{Multicolored Clique}
to \textsc{CCS-R}, which we believe to be of independent interest. 
We can assume, without loss of generality, that for an input $(G,c,k)$ of \textsc{Multicolored Clique}, $G$ is connected 
and $c$ is a proper vertex-coloring, i.e., for any two distinct vertices $u,v\in V(G)$ with $c(u)=c(v)$ we have $\{u,v\}\notin E(G)$. 
Before we proceed let us define a graph operation. 

\begin{definition}
Let $G$ be a graph and let $c\colon V(G)\rightarrow\{1,\ldots, k\}$ be a proper vertex coloring of $V(G)$. Let $H$ be a graph on the vertex set
$\{1,\ldots,k\}$. We define the graph~$G\upharpoonright_c H$ as follows. 
We remove all edges $\{u,v\}\in E(G)$
such that $c(u)=i$ and $c(v)=j$ and $\{i,j\}\not\in E(H)$. 
We subdivide every remaining edge, i.e.\ for every remaining edge $\{u,v\}$
we introduce a new vertex $s_{uv}$, remove the edge $\{u,v\}$ and
introduce instead the two edges $\{u,s_{uv}\}$ and $\{v,s_{uv}\}$. We write
$W(G\upharpoonright_c H)$ for the set of all subdivision vertices $s_{uv}$ (see Figure~\ref{fig:Gcomp}).
\end{definition}

\begin{figure}[!ht]
    \centering
    \begin{subfigure}[b]{0.30\textwidth}
       \begin{tikzpicture}[scale=0.83]
 \draw[fill=yellow!30] (-2.5,0) ellipse (1cm and 0.35cm);
 \draw[fill=red!30] (-2.5,-2) ellipse (1cm and 0.35cm);
 \draw[fill=green!30] (-2.25,-4) ellipse (0.7cm and 0.35cm);
 \draw[fill=blue!20] (-2.5,-6) ellipse (1cm and 0.35cm);
 \node[]  at (-3.1,0) (a) {$\bullet$};
 \node[]  at (-2.5,0) (a) {$\bullet$};

\node at (-2.3,0) {$u$}; 
 
 \node[]  at (-1.75,0) (a) {$\bullet$};
 \node[]  at (-3,-2) (a) {$\bullet$};

\node at (-2.8,-2) {$v$}; 
 
 \node[]  at (-2,-2) (a) {$\bullet$};
\node[]  at (-2.5,-4) (a) {$\bullet$};
 \node[]  at (-3.2,-6) (a) {$\bullet$};
 \node[]  at (-2,-6) (a) {$\bullet$};
\draw (-2.5,0) to[out=190,in=90] (-3,-2);
\draw (-2.5,0)--(-2,-2);
\draw (-1.75,0)--(-2,-2);
\draw (-2.5,-4)--(-2,-2);       
\draw (-2.5,-4)--(-2,-6);
\draw (-2.5,-4) to[out=0,in=-80] (-1.75,0);
\draw (-2,-6) to[out=40,in=-70] (-1.75,0);
\draw (-2,-6) to[out=180,in=180] (-2,-2);
\draw (-3.2,-6) to[out=100,in=220] (-3.1,0);
 \node[]  at (-5,0) (a) {$1$};
 \node[]  at (-5,-2) (a) {$2$};
\node[]  at (-5,-4) (a) {$3$};
 \node[]  at (-5,-6) (a) {$4$};
       \end{tikzpicture}
        \caption{A graph $G$ and a proper coloring $c \colon V(G)\rightarrow \{1,\ldots,4\}$}
        \label{fig:G}
    \end{subfigure}~~~~~~ 
    \begin{subfigure}[b]{0.22\textwidth}
       \begin{tikzpicture}[scale=0.83]
 \node[]  at (-2.5,0) (a) {$\bullet$};
 \node[]  at (-2,-2) (a) {$\bullet$};
\node[]  at (-2.5,-4) (a) {$\bullet$};
 \node[]  at (-2,-6) (a) {$\bullet$};
\draw (-2.5,0)--(-2,-2);
\draw (-2,-6)--(-2,-2);
\draw (-2.5,-4)--(-2,-2);       
\node[]  at (-2.75,0) (a) {$1$};
 \node[]  at (-2.25,-2) (a) {$2$};
\node[]  at (-2.75,-4) (a) {$3$};
 \node[]  at (-2.25,-6) (a) {$4$};
       \end{tikzpicture}
        \caption{A graph $H$ on the vertex set $\{1,\ldots,4\}$}
        \label{fig:H}
    \end{subfigure}~~~~~~
        \begin{subfigure}[b]{0.33\textwidth}
       \begin{tikzpicture}[scale=0.83]
 \draw[fill=yellow!30] (-2.5,0) ellipse (1cm and 0.35cm);
 \draw[fill=red!30] (-2.5,-2) ellipse (1cm and 0.35cm);
 \draw[fill=green!30] (-2.25,-4) ellipse (0.7cm and 0.35cm);
 \draw[fill=blue!20] (-2.5,-6) ellipse (1cm and 0.35cm);
 \node[]  at (-3.1,0) (a) {$\bullet$};
 \node[]  at (-2.5,0) (a) {$\bullet$};

\node at (-2.3,0) {$u$}; 
 
 \node[]  at (-1.75,0) (a) {$\bullet$};
 \node[]  at (-3,-2) (a) {$\bullet$};
 
 \node at (-2.8,-2) {$v$};  
 
 \node[]  at (-2,-2) (a) {$\bullet$};
\node[]  at (-2.5,-4) (a) {$\bullet$};
 \node[]  at (-3.2,-6) (a) {$\bullet$};
 \node[]  at (-2,-6) (a) {$\bullet$};
\draw (-2.5,0) to[out=190,in=90] (-3,-2);
\draw (-2.5,0)--(-2,-2);
\draw (-1.75,0)--(-2,-2);
\draw (-2.5,-4)--(-2,-2);       
\draw (-2,-6) to[out=180,in=180] (-2,-2);

\node[]  at (-3.05,-1) (a) {$\bullet$};
 \node[]  at (-2.25,-1) (a) {$\bullet$};
 \node[]  at (-1.87,-1) (a) {$\bullet$};
 \node[]  at (-3.17,-4) (a) {$\bullet$};
 \node[]  at (-2.25,-3) (a) {$\bullet$};

\node[]  at (-3.95,-1) (a) {$s_{uv}=w_1$};
 \node[]  at (-2.55,-1) (a) {$w_2$};
 \node[]  at (-1.5,-1) (a) {$w_3$};
 \node[]  at (-3.45,-4) (a) {$w_4$};
 \node[]  at (-1.9,-3) (a) {$w_5$};

       \end{tikzpicture}
        \caption{The graph $G\upharpoonright_c H$. Here, $W(G\upharpoonright_c H)=\{w_1, \ldots,w_5\}$}
        \label{fig:GtoH}
    \end{subfigure} 
    \caption{Construction of $G\upharpoonright_c H$.}\label{fig:Gcomp}
\end{figure}

Let $(G,c,k)$ be the input instance of \textsc{Multicolored Clique}, where $G$ is a connected 
graph and $c$ is a proper $k$-vertex-coloring of $G$. We construct an instance 
$(H,\widehat{c}\colon V(H)\mapsto [k+1], Q_s, Q_t, 2k)$  of \textsc{CCS-R}. 

We first construct a routing gadget. For $1\leq i\leq k$, let $T^i$ be the star with vertex set $\{1,\ldots, k\}$ having vertex $i$ as the center. 
For any $1\leq i\leq k$ and $1\leq r\leq 20k$, we let~$H^{(i,r)}$ be a copy of the graph
$G\upharpoonright_c T^i$. We let~$c_{(i,r)}$ be the the partial vertex-coloring of~$H^{(i,r)}$ that is naturally inherited from $G$. 
For an illustration, consider the input instance $(G,c,k)$ of \textsc{Multicolored Clique} depicted in Figure~\ref{fig:G}. 
Then, $T^2$ is identical to the graph $H$ in Figure~\ref{fig:H} and Figure~\ref{fig:GtoH} represents $H^{(2,r)}=G\upharpoonright_c T^2$, for any $1\leq r\leq 20k$.
Now, for $1\leq i\leq k$ we define a graph $H^i$ as follows. 
We use $W(H^{(i,r)})$ to denote the set of subdivision vertices in $H^{(i,r)}$. 
For $1\leq r<20k$ and all vertices $u,v$ in $V(H^{(i,r)})\setminus W(H^{(i,r)})$, we connect the copy of 
the subdivision vertex~$s_{uv}$ in $H^{(i,r)}$ (if it exists) with the 
copies of the vertices~$u$ and $v$ in~$H^{(i,r+1)}$ 
(see Figure~\ref{fig:ConstructionHi} for an illustration of a portion of $H^1$). 
We use $W(H^i)$ to denote the set of subdivision vertices $\bigcup_{r\in [20k]} W(H^{(i,r)})$. 

For each $1\leq i\leq k$, we use $c_i$ to denote a coloring on $V(H^i)$ that is a union of 
$c_{(i,1)},c_{(i,2)},\ldots,c_{(i,20k)}$ and we color all the copies of the subdivision vertices using a new color $k+1$. 
In other words, we know that for each $u\in V(H^i)$ we have  $u\in V(H^{(i,r)})$, for 
some $r\in \{1,\ldots,20k\}$. Hence, if $u\in V(H^{(i,r)}) \setminus W(H^{(i,r)})$ then we set $c_i(u)=c_{(i,r)}(u)$. 
For all $s_{uv}\in W(H^i)$, we set $c_i(s_{uv})=k+1$.

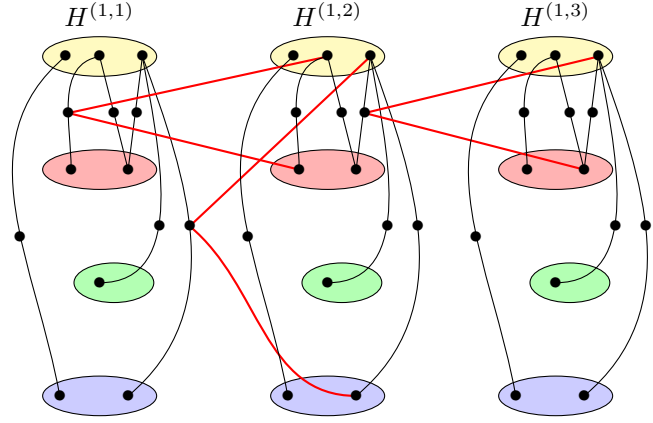
\begin{figure}[t]
    \centering
       \begin{tikzpicture}[scale=0.75]

\node[]  at (-2.5,0.75) (a) {$H^{(1,1)}$};
\node[]  at (1.5,0.75) (a) {$H^{(1,2)}$};
\node[]  at (5.5,0.75) (a) {$H^{(1,3)}$};

 \draw[fill=yellow!30] (-2.5,0) ellipse (1cm and 0.35cm);
 \draw[fill=red!30] (-2.5,-2) ellipse (1cm and 0.35cm);
 \draw[fill=green!30] (-2.25,-4) ellipse (0.7cm and 0.35cm);
 \draw[fill=blue!20] (-2.5,-6) ellipse (1cm and 0.35cm);

\draw[fill=yellow!30] (1.5,0) ellipse (1cm and 0.35cm);
 \draw[fill=red!30] (1.5,-2) ellipse (1cm and 0.35cm);
 \draw[fill=green!30] (1.75,-4) ellipse (0.7cm and 0.35cm);
 \draw[fill=blue!20] (1.5,-6) ellipse (1cm and 0.35cm);

 \draw[fill=yellow!30] (5.5,0) ellipse (1cm and 0.35cm);
 \draw[fill=red!30] (5.5,-2) ellipse (1cm and 0.35cm);
 \draw[fill=green!30] (5.75,-4) ellipse (0.7cm and 0.35cm);
 \draw[fill=blue!20] (5.5,-6) ellipse (1cm and 0.35cm);

 \node[]  at (-3.1,0) (a) {$\bullet$};
 \node[]  at (-2.5,0) (a) {$\bullet$};
 \node[]  at (-1.75,0) (a) {$\bullet$};
 \node[]  at (-3,-2) (a) {$\bullet$};
 \node[]  at (-2,-2) (a) {$\bullet$};
\node[]  at (-2.5,-4) (a) {$\bullet$};
 \node[]  at (-3.2,-6) (a) {$\bullet$};
 \node[]  at (-2,-6) (a) {$\bullet$};
\draw (-2.5,0) to[out=190,in=90] (-3,-2);
\draw (-2.5,0)--(-2,-2);
\draw (-1.75,0)--(-2,-2);
\draw (-2.5,-4) to[out=0,in=-80] (-1.75,0);
\draw (-2,-6) to[out=40,in=-70] (-1.75,0);
\draw (-3.2,-6) to[out=100,in=220] (-3.1,0);


\draw[thick, red] (-3.05,-1) --(1.5,0);
\draw[thick, red] (-3.05,-1)  --(1,-2);
\draw[thick, red] (-0.92,-3) --(2.25,0);
\draw[thick, red] (-0.92,-3) to[out=-40,in=180] (2,-6);

\node[]  at (-3.05,-1) (a) {$\bullet$};
 \node[]  at (-2.25,-1) (a) {$\bullet$};
 \node[]  at (-1.85,-1) (a) {$\bullet$};
 \node[]  at (-1.45,-3) (a) {$\bullet$};
 \node[]  at (-0.92,-3) (a) {$\bullet$};
\node[]  at (-3.9,-3.2) (a) {$\bullet$};

  \node[]  at (0.9,0) (a) {$\bullet$};
 \node[]  at (1.5,0) (a) {$\bullet$};
 \node[]  at (2.25,0) (a) {$\bullet$};
 \node[]  at (1,-2) (a) {$\bullet$};
 \node[]  at (2,-2) (a) {$\bullet$};
\node[]  at (1.5,-4) (a) {$\bullet$};
 \node[]  at (0.8,-6) (a) {$\bullet$};
 \node[]  at (2,-6) (a) {$\bullet$};
\draw (1.5,0) to[out=190,in=90] (1,-2);
\draw (1.5,0)--(2,-2);
\draw (2.25,0)--(2,-2);
\draw (1.5,-4) to[out=0,in=-80] (2.25,0);
\draw (2,-6) to[out=40,in=-70] (2.25,0);
\draw (0.8,-6) to[out=100,in=220] (0.9,0);


\draw[thick, red] (2.15,-1) --(6.25,0);
\draw[thick, red] (2.15,-1)  --(6,-2);

\node[]  at (0.95,-1) (a) {$\bullet$};
 \node[]  at (1.75,-1) (a) {$\bullet$};
 \node[]  at (2.15,-1) (a) {$\bullet$};

\node[]  at (0.1,-3.2) (a) {$\bullet$};

 \node[]  at (2.55,-3) (a) {$\bullet$};
 \node[]  at (3.08,-3) (a) {$\bullet$};

 \node[]  at (4.9,0) (a) {$\bullet$};
 \node[]  at (5.5,0) (a) {$\bullet$};
 \node[]  at (6.25,0) (a) {$\bullet$};
 \node[]  at (5,-2) (a) {$\bullet$};
 \node[]  at (6,-2) (a) {$\bullet$};
\node[]  at (5.5,-4) (a) {$\bullet$};
 \node[]  at (4.8,-6) (a) {$\bullet$};
 \node[]  at (6,-6) (a) {$\bullet$};
\draw (5.5,0) to[out=190,in=90] (5,-2);
\draw (5.5,0)--(6,-2);
\draw (6.25,0)--(6,-2);
\draw (5.5,-4) to[out=0,in=-80] (6.25,0);
\draw (6,-6) to[out=40,in=-70] (6.25,0);
\draw (4.8,-6) to[out=100,in=220] (4.9,0);

\node[]  at (4.95,-1) (a) {$\bullet$};
 \node[]  at (5.75,-1) (a) {$\bullet$};
 \node[]  at (6.15,-1) (a) {$\bullet$};



\node[]  at (4.1,-3.2) (a) {$\bullet$};

 \node[]  at (6.55,-3) (a) {$\bullet$};
 \node[]  at (7.08,-3) (a) {$\bullet$};

 \end{tikzpicture}
    \caption{Construction of $H^1$ from  the instance $(G,c)$ depicted in Figure~\ref{fig:G}. The red edges are some of the ``crossing'' edges but not all of them. 
%
}
\label{fig:ConstructionHi}
\end{figure}

Now, define a graph~$R$, which is super graph of $H^1\cup \ldots \cup H^{k}$, as follows. 
For $1\leq i<k$ and all vertices $u$ and $v$, we connect the copy of 
the subdivision vertex $s_{uv}$ in~$H^{(i,20k)}$ (if it exists) with the 
copies of the vertices $u$ and $v$ in~$H^{(i+1,1)}$ (see Figure~\ref{fig:ConstructionR} for an illustration).

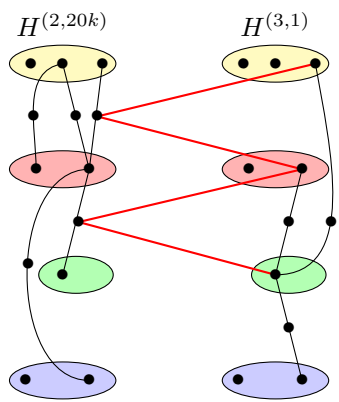
\begin{figure}[ht]
    \centering
       \begin{tikzpicture}[scale=0.7]

\node[]  at (1.5,0.75) (a) {$H^{(2,20k)}$};
\node[]  at (5.5,0.75) (a) {$H^{(3,1)}$};


\draw[fill=yellow!30] (1.5,0) ellipse (1cm and 0.35cm);
 \draw[fill=red!30] (1.5,-2) ellipse (1cm and 0.35cm);
 \draw[fill=green!30] (1.75,-4) ellipse (0.7cm and 0.35cm);
 \draw[fill=blue!20] (1.5,-6) ellipse (1cm and 0.35cm);

 \draw[fill=yellow!30] (5.5,0) ellipse (1cm and 0.35cm);
 \draw[fill=red!30] (5.5,-2) ellipse (1cm and 0.35cm);
 \draw[fill=green!30] (5.75,-4) ellipse (0.7cm and 0.35cm);
 \draw[fill=blue!20] (5.5,-6) ellipse (1cm and 0.35cm);

  \node[]  at (0.9,0) (a) {$\bullet$};
 \node[]  at (1.5,0) (a) {$\bullet$};
 \node[]  at (2.25,0) (a) {$\bullet$};
 \node[]  at (1,-2) (a) {$\bullet$};
 \node[]  at (2,-2) (a) {$\bullet$};
\node[]  at (1.5,-4) (a) {$\bullet$};
 \node[]  at (0.8,-6) (a) {$\bullet$};
 \node[]  at (2,-6) (a) {$\bullet$};
\draw (1.5,0) to[out=190,in=90] (1,-2);
\draw (1.5,0)--(2,-2);
\draw (2.25,0)--(2,-2);
\draw (1.5,-4)--(2,-2);       
\draw (2,-6) to[out=180,in=180] (2,-2);


\draw[thick, red] (2.15,-1) --(6.25,0);
\draw[thick, red] (2.15,-1)  --(6,-2);
\draw[thick, red] (1.8,-3) --(6,-2);
\draw[thick, red] (1.8,-3) --(5.5,-4);

\node[]  at (0.95,-1) (a) {$\bullet$};
 \node[]  at (1.75,-1) (a) {$\bullet$};
 \node[]  at (2.15,-1) (a) {$\bullet$};
 \node[]  at (0.85,-3.8) (a) {$\bullet$};
\node[]  at (1.8,-3) (a) {$\bullet$};

 \node[]  at (4.9,0) (a) {$\bullet$};
 \node[]  at (5.5,0) (a) {$\bullet$};
 \node[]  at (6.25,0) (a) {$\bullet$};
 \node[]  at (5,-2) (a) {$\bullet$};
 \node[]  at (6,-2) (a) {$\bullet$};
\node[]  at (5.5,-4) (a) {$\bullet$};
 \node[]  at (4.8,-6) (a) {$\bullet$};
 \node[]  at (6,-6) (a) {$\bullet$};
 
\node[]  at (5.75,-3) (a) {$\bullet$};
\node[]  at (6.55,-3) (a) {$\bullet$};
\node[]  at (5.75,-5) (a) {$\bullet$};

\draw (5.5,-4)--(6,-2);       
\draw (5.5,-4)--(6,-6);
\draw (5.5,-4) to[out=0,in=-80] (6.25,0);

 \end{tikzpicture}
    \caption{Illustration of the subgraph of $R$ induced on $V(H^{(2,20k)})\cup V(H^{3,1})$ constructed from  the instance $(G,c,k)$ depicted in Figure~\ref{fig:G}. The red edge are some of the ``crossing edges''. 
    }\label{fig:ConstructionR}
\end{figure}

We additionally introduce two subgraphs $H^0$ and 
$H^{k+1}$. The graph $H^0$ is obtained by subdividing each edge of a star on vertex set $\{v_1,\ldots,v_k\}$ 
centered at $v_1$. Here we use $w_2,\ldots,w_k$ to denote the subdivision vertices. 
Similarly, the graph $H^{k+1}$ is obtained by subdividing each edge of star on $\{x_1,\ldots,x_k\}$ 
centered at $x_k$. Here $y_1,\ldots,y_{k-1}$ denote the subdivision vertices.
%
%
%
Let $c_0$ and $c_{k+1}$ be the colorings 
on $\{v_1,\ldots, v_k,w_2,\ldots,w_{k}\}$ and $\{x_1,\ldots, x_k,y_1,\ldots,y_{k-1}\}$, respectively, 
defined as follows.  For all $1\leq i\leq k$, \mbox{$c_0(v_i)=i$} and $c_{k+1}(x_i)=i$. 
For all $2\leq i\leq k$, $c_0(w_i)=k+1$ and for all $1\leq i\leq k-1$, $c_{k+1}(y_i)=k+1$. 
Observe that we may interpret $H^0$ as $K[\{v_1,\ldots,v_k\}]\upharpoonright_{c_{0}} T^0$ and $H^{k+1}$ as $K[\{x_1,\ldots,x_k\}]\upharpoonright_{c_{k+1}} T^{k+1}$, where 
$T^0$ and $T^{k+1}$ are two trees on vertex 
set $\{1,\ldots,k\}$, with $E(T^0)=\{\{1,i\}\colon 2\leq i\leq k\}$ and $E(T^{k+1})=\{\{k,i\}\colon 1\leq i\leq k-1\}$. 

Finally, for each $2\leq i\leq k$, we connect 
the ``subdivision vertex''~$w_i$ (adjacent to $v_1$ and $v_i$) to all vertices $v \in V(H^{(1,1)})$ colored $1$ or $i$, i.e., with $c_{(1,1)}(v)\in \{1,i\}$. 
For each subdivision vertex \mbox{$s_{ab}\in W(H^{(k,20k)})$}, we connect 
$s_{ab}$ to $x_k$ and $x_i$, where $k=c_k(a)=c_{(k,20k)}(a)$ and $i=c_k(b)=c_{(k,20k)}(b)$.  
Recall that $s_{ab}$ is adjacent a vertex of color $k$ and a vertex of color $i$, for some $i < k$.  
This completes the construction of $H$ (see Figure~\ref{fig:ConstructionHR}). We define 
$\widehat{c}\colon V(H)\mapsto [k+1]$ to be the union of $c_0,\ldots,c_{k+1}$.

%
\begin{figure}[h]
    \centering
       \begin{tikzpicture}[scale=0.75]
 
  \draw[fill=yellow!30] (-2.5,0) ellipse (1cm and 0.35cm);
 \draw[fill=red!30] (-2.5,-2) ellipse (1cm and 0.35cm);
 \draw[fill=green!30] (-2.25,-4) ellipse (0.7cm and 0.35cm);
 \draw[fill=blue!20] (-2.5,-6) ellipse (1cm and 0.35cm);

 \draw[fill=yellow!30] (4.5,0) ellipse (1cm and 0.35cm);
 \draw[fill=red!30] (4.5,-2) ellipse (1cm and 0.35cm);
 \draw[fill=green!30] (4.75,-4) ellipse (0.7cm and 0.35cm);
 \draw[fill=blue!20] (4.5,-6) ellipse (1cm and 0.35cm);

 \draw[thick,red]  (-3.1,0)-- (-5.5,-1)--(-3,-2);
\draw[thick,red]  (-5.65,-3)to [out=30,in=-120](-3.1,0); 
\draw[thick,red]  (-5.65,-3)--(-2.5,-4);
\draw[thick,red] (-6.15,-5)--(-3.2,-6);
\draw[thick,red]  (-6.15,-5) to [out=30,in=-100](-3.1,0);
 
 \draw[thick,red]  (-2.5,0) -- (-5.5,-1);

  \node[]  at (-6,0) (a) {$\bullet$};
 \node[]  at (-5,-2) (a) {$\bullet$};
\node[]  at (-5.5,-4) (a) {$\bullet$};
 \node[]  at (-6.2,-6) (a) {$\bullet$};

  \node[]  at (-5.5,-1) (a) {$\bullet$};
 \node[]  at (-5.65,-3) (a) {$\bullet$};
\node[]  at (-6.15,-5) (a) {$\bullet$};

\draw (-6,0)--(-6.2,-6);
\draw (-6,0)--(-5,-2);
\draw (-6,0)--(-5.5,-4);

\node[]  at (-6.3,0) (a) {$v_1$};
 \node[]  at (-5.3,-2) (a) {$v_2$};
\node[]  at (-5.8,-4) (a) {$v_3$};
 \node[]  at (-6.5,-6) (a) {$v_4$};
      
  \node[]  at (-5.65,-1.25) (a) {$w_2$};
 \node[]  at (-5.9,-3.1) (a) {$w_3$};
\node[]  at (-6.5,-5) (a) {$w_4$};
       

\node[]  at (-2.5,0.75) (a) {$H^{(1,1)}$};
\node[]  at (5.5,0.75) (a) {$H^{(4,20k)}$};

 \node[]  at (-3.1,0) (a) {$\bullet$};
 \node[]  at (-2.5,0) (a) {$\bullet$};
 \node[]  at (-1.75,0) (a) {$\bullet$};
 \node[]  at (-3,-2) (a) {$\bullet$};
 \node[]  at (-2,-2) (a) {$\bullet$};
\node[]  at (-2.5,-4) (a) {$\bullet$};
 \node[]  at (-3.2,-6) (a) {$\bullet$};
 \node[]  at (-2,-6) (a) {$\bullet$};
\draw (-2.5,0) to[out=190,in=90] (-3,-2);
\draw (-2.5,0)--(-2,-2);
\draw (-1.75,0)--(-2,-2);
\draw (-2.5,-4) to[out=0,in=-80] (-1.75,0);
\draw (-2,-6) to[out=40,in=-70] (-1.75,0);
\draw (-3.2,-6) to[out=100,in=220] (-3.1,0);

\node[]  at (-3.05,-1) (a) {$\bullet$};
 \node[]  at (-2.25,-1) (a) {$\bullet$};
 \node[]  at (-1.85,-1) (a) {$\bullet$};
 \node[]  at (-1.45,-3) (a) {$\bullet$};
 \node[]  at (-0.92,-3) (a) {$\bullet$};
\node[]  at (-3.9,-3.2) (a) {$\bullet$};


\draw[thick,red] (6.08,-3)--(8,0);
\draw[thick,red] (6.08,-3)to [out=0,in=110](8,-6);

\draw[thick,red] (4.75,-5)--(6.8,-4);
\draw[thick,red] (4.75,-5)--(8,-6);

 \node[]  at (3.9,0) (a) {$\bullet$};
 \node[]  at (4.5,0) (a) {$\bullet$};
 \node[]  at (5.25,0) (a) {$\bullet$};
 \node[]  at (4,-2) (a) {$\bullet$};
 \node[]  at (5,-2) (a) {$\bullet$};
\node[]  at (4.5,-4) (a) {$\bullet$};
 \node[]  at (3.8,-6) (a) {$\bullet$};
 \node[]  at (5,-6) (a) {$\bullet$};
\draw (4.5,-4)--(5,-6);
\draw (5,-6) to[out=40,in=-70] (5.25,0);
\draw (5,-6) to[out=180,in=180] (5,-2);
\draw (3.8,-6) to[out=100,in=220] (3.9,0);

 
 \node[]  at (4.75,-5) (a) {$\bullet$};

\node[]  at (3.1,-3.2) (a) {$\bullet$};

 \node[]  at (4,-3) (a) {$\bullet$};


 \node[]  at (6.08,-3) (a) {$\bullet$};

        
  \node[]  at (8,0) (a) {$\bullet$};
 \node[]  at (7.25,-2) (a) {$\bullet$};
\node[]  at (6.75,-4) (a) {$\bullet$};
 \node[]  at (8,-6) (a) {$\bullet$};

  \node[]  at (8,-1) (a) {$\bullet$};
 \node[]  at (7.45,-3) (a) {$\bullet$};
\node[]  at (7.35,-5) (a) {$\bullet$};

\draw (8,0)--(8,-6);
\draw (7.25,-2)--(8,-6)--(6.75,-4);

\node[]  at (8.3,0) (a) {$x_1$};
 \node[]  at (7.6,-2) (a) {$x_2$};
\node[]  at (6.7,-4.25) (a) {$x_3$};
 \node[]  at (8.3,-6) (a) {$x_4$};
      
  \node[]  at (8.3,-1) (a) {$y_1$};
 \node[]  at (7.8,-3) (a) {$y_2$};
\node[]  at (7.3,-5.3) (a) {$y_3$};

 \end{tikzpicture}
    \caption{Illustration of connection between $H^0$ and $R$, and $H^{k+1}$ and $R$  from the instance $(G,c,k)$ depicted in Figure~\ref{fig:G}. The red edge are some of the ``crossing edges'' between $H^0$ and $H^1$, and $H^{k}$ and $H^{k+1}$.} 

\label{fig:ConstructionHR}
\end{figure}
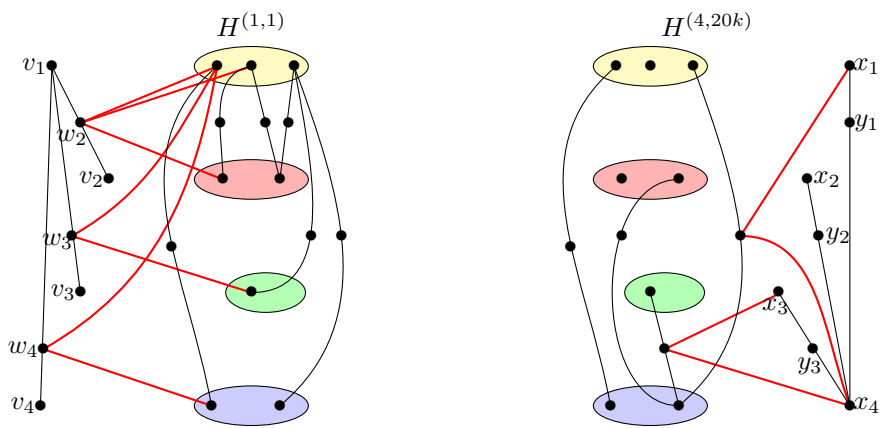


\begin{observation}
The sets $\{ v_1,\ldots, v_k, w_2,\ldots, w_{k}\}$
and $\{x_1,\ldots, x_k$, $y_1,\ldots, y_{k-1}\}$ 
are solutions of size $2k - 1$ of the \textsc{CCS} instance $(H,\widehat{c},2k)$. 
\end{observation}


We define the starting configuration $Q_s$ as the set 
$\{v_1,\ldots, v_k, w_2,\ldots, w_{k}\}$
and the target configuration $Q_t$ as the set 
$\{x_1,\ldots, x_k$, $y_1,\ldots, y_{k-1}\}$. 
We now consider the instance $(H,\widehat{c},Q_s,Q_t,2k)$ of the 
\textsc{CCS-R} problem. That is, the bound on the sizes of the solutions in the reconfiguration sequence is at most $2k$. 
Before we analyze the reconfiguration properties of $H$, let us
verify that $H$ is $4$-degenerate. 

\begin{lemma}[$\star$]
\label{lem:fourdeg}
The graph $H$ is $4$-degenerate. 
\end{lemma}


\begin{lemma}[$\star$]\label{lem:forward}
if there exists a $k$-colored clique in $G$ then there is reconfiguration sequence of length $\Oof(k^3)$ from $Q_s$ to $Q_t$ in $(H,\widehat{c},2k)$.
\end{lemma}

\begin{proof}[Proof sketch]
We aim
to shift the connected vertices of $Q_s$ through the subgraphs
$H^1,\ldots, H^k$ (in that order) to maintain connectivity and eventually shift all the tokens to $Q_t$.  
For each $u_i\in V(G)$, $1\leq j\leq k$ and $1\leq r\leq 20k$, we use  $u_i^{(j,r)}$ to denote the copy of $u_i$ in $H^{(j,r)}$. 
Let $C=\{u_1,\ldots,u_k\}$ be a $k$-colored clique in $G$ such that $c(u_i)=i$, for all $1\leq i\leq k$. 
To prove the lemma, we need to define a reconfiguration sequence starting from $Q_s$ and ending at $Q_t$ 
such that the cardinality of any solution in the sequence is at most $2k$. 
First we define $k$ ``colored'' trees $\widehat{T}_1,\ldots, \widehat{T}_k$ each on $2k-1$ vertices, and then 
prove that there are reconfiguration sequences from $Q_s$ to $V(\widehat{T}_1)$,  $V(\widehat{T}_i)$ 
to $V(\widehat{T}_{i+1})$ for all $1\leq i<k$, and $V(\widehat{T}_k)$ to $Q_t$. 

We start by defining $\widehat{T}_1,\ldots, \widehat{T}_k$. 
For each $1\leq i\leq k$, $C_i=\{u^{(i,1)}_1,\ldots,u_k^{(i,1)}\}$ and $S_i=\{z\in V(H^{(i,1)})\colon N_{H^{(i,1)}}(z)\cap C_i=2\}$. 
That is, for each $1\leq j\leq k$ and $j\neq i$,  $s_{u_i^{(i,1)}u_j^{(i,1)}}\in S_i$ (the subdivision vertex on the edge $u_i^{(i,1)}u_j^{(i,1)}$ is in $S_i$), 
and $\vert S_i\vert =k-1$. In other words, $C_i$ contains the copies of the vertices of the clique $C$ in $H^{(i,1)}$ and $S_i$ contains subdivision vertices corresponding to $k-1$ edges in the clique incident on the $i$th colored vertex of the clique, such that $H[C_i\cup S_i]$ is a tree. Now, define $\widehat{T}_i=H[C_i\cup S_i]$. 
It is easy to verify that $\widehat{c}(C_i\cup S_i)=\{1,\ldots,k+1\}$ and hence $C_i\cup S_i=V(\widehat{T}_i)$ 
is a solution to the \textsc{CCS} instance $(H,\widehat{c},2k)$. 
Let $T_s=H[Q_s]$ and $T_t=H[Q_t]$. Note that $T_s$ and $T_t$ are trees on $2k-1$ vertices. 

\medskip
\noindent{\bf Case 1:  Reconfiguration from $Q_s$ to $V(\widehat{T}_1)$.} 
Informally, we move to $\widehat{T}_1$ by adding a token on $u_i^{(1,1)}$ and then removing tokens from $v_i$ for $i$ in the order  
$2,\ldots,k,1$ (for a total of $2k$ token additions/removals). Finally, we move the tokens from $\{w_2,\ldots,w_{k-1}\}$ to $S_1$ in $2(k-1)$ steps.  
The length of the reconfiguration sequence is $2k+2(k-1)=4k-2$.

\medskip
\noindent{\bf Case 2:  Reconfiguration from $V(\widehat{T}_i)$ to $V(\widehat{T}_{i+1})$.} First we define $20k$ trees $P_1,\ldots P_{20k}$, each on $2k-1$ vertices 
such that for all $1\leq r\leq 20k$, $(i)$ $V(P_r)\subseteq V(H^{(i,r)})$, and $(ii)$ $\widehat{T}_i=P_1$. 
Then we give a reconfiguration sequence from $V(P_r)$ to $V(P_{r+1})$ for all $r\in [20k-1]$ and a reconfiguration sequence from $V(P_{20k})$ to $V(\widehat{T}_{i+1})$.  

Recall that $C=\{u_1,\ldots,u_k\}$ is a $k$-colored clique in $G$ such that $c(u_i)=i$ for all $1\leq i\leq k$. 
For $1\leq r\leq 20k$, let $C_i^r=\{u^{(i,r)}_1,\ldots,u_k^{(i,r)}\}$ and \mbox{$S_i^r=\{z\in V(H^{(i,r)})\colon N_{H^{(i,r)}}(z)\cap C_i^r=2\}$}. 
That is, for each $1\leq j\leq k$ and $j\neq i$,  $s_{u_i^{(i,r)}u_j^{(i,r)}}\in S_i^r$ 
(i.e, the subdivision vertex on the edge $u_i^{(i,r)}u_j^{(i,r)}$ is in $S_i^r$) and $\vert S_i^r\vert =k-1$. 
Let $P_r=H[C_i^r\cup S^r_i]$. 
Notice that  for all $r\in [20k]$, $P_r$
is a tree on $2k-1$ vertices.   
Moreover, for each $1\leq r\leq 20k$, $V(P_r)$ is a solution to the \textsc{CCS} instance $(H,\widehat{c},2k)$. 
%
%
By arguments similar to those given for Case $1$, one can prove that  there is a reconfiguration sequence of length  
$4k-2$ from $V(P_r)$ to $V(P_{r+1})$, for all $1\leq r<20k$.

For the reconfiguration sequence from $V(P_{20k})$ to $V(\widehat{T}_{i+1})$ we refer the reader to the complete proof in the appendix.

\medskip
\noindent{\bf Case 3:  Reconfiguration from $V(\widehat{T}_k)$ to $V({T}_t)$.} 
The arguments for this case are similar to those given in Case 1, we therefore omit the details.  
\end{proof}

\begin{lemma}[$\star$]
\label{lem:backward}
If there is a reconfiguration sequence  from $Q_s$ to $Q_t$ then there is 
a $k$-colored clique in $G$. 
\end{lemma}

\begin{theorem}\label{thm:hardness}
\textsc{CCS-R} parameterized by $k+\ell$ is $\mathsf{W}[1]$-hard on $4$-degenerate graphs.
\end{theorem}

\subparagraph*{Reduction from CCS-R to CDS-R.} We give a polynomial-time parameter-preserving reduction 
from \textsc{CCS-R} to \textsc{CDS-R} that is fairly straightforward. 
Let $(G,c,S,T,k)$ be an instance of \textsc{CCS-R}. Let $c \colon V(G)\mapsto \{1,\ldots,k'\}$, where $k'\leq k$. 
We construct a graph $H$ as follows. For each $1\leq i\leq k'$, we add a vertex $d_i$ and connect $d_i$ to all the vertices in $c^{-1}(i)$. 
Next, for each $1\leq i\leq k'$, we add $2k+1$ pendent vertices to $d_i$. 
That is, we add vertices $\{x_{i,j}\colon 1\leq i\leq k', 1\leq j\leq 2k+1\}$ and edges $\{\{x_{i,j},d_i\}\colon 1\leq i\leq k', 1\leq j\leq 2k+1\}$. 
Let $D=\{d_1,\ldots,d_{k'}\}$. We output $(H,S\cup D, T\cup D,k+k')$ as the new \textsc{CDS-R} instance.  

\begin{lemma}
If $G$ is a $d$-degenerate graph then $H$ is a $(d+1)$-degenerate graph. 
\end{lemma}

\begin{proof}
For each vertex $v\in V(G)$, $d_H(v)=d_{G}(v)+1$. Thus, after removing $V(G)$ and $\{x_{i,j}\colon 1\leq i\leq k', 1\leq j\leq 2k+1\}$, the remaining graph is edgeless. 
\end{proof}

It is easy to verify that for any reconfiguration sequence $S=S_1,\ldots, S_{\ell}=T$ of $(G,c,S,T,k)$, 
$S\cup D=S_1\cup D,\ldots, S_{\ell}\cup D=T\cup D$ is a reconfiguration sequence of $(H,S\cup D, T\cup D,k+k')$. 
Now we prove the reverse direction. 

\begin{lemma}
If $(H,S\cup D, T\cup D,k+k')$ is a \yes-instance then  $(G,c,S,T,k)$ is a \yes-instance.
\end{lemma}
\begin{proof}
Let $Z$ be a dominating set in $H$ of size at most $k+k'$. Then $D\subseteq H$.  Moreover for any minimal connected dominating set $Z$ in $H$, $Z\cap \{x_{i,j}\colon 1\leq i\leq k', 1\leq j\leq 2k+1\}=\emptyset$, $H[Z\setminus D]$ is connected, and $Z\setminus D$ contains a vertex from  $c^{-1}(i)$ for all $1\leq i\leq k'$ (recall that~$G$ is a subgraph of $H$). Therefore, by deleting $D$ from each set in a reconfiguration sequence of $(H,S\cup D, T\cup D,k+k')$, we get a valid reconfiguration sequence of $(G,c,S,T,k)$. This completes the proof. 
\end{proof}

Thus, by Theorem~\ref{thm:hardness}, we have the following theorem.

\begin{theorem}\label{thm:hardnessfinal}
\textsc{CDS-R} parameterized by $k+\ell$ is $\mathsf{W}[1]$-hard on $5$-degenerate graphs.
\end{theorem}

\section{Fixed-parameter tractability on planar graphs}\label{sec:planar}

This section is devoted to proving that \textsc{CDS-R} under \textsf{TAR} parameterized by
$k$ is fixed-parameter tractable on planar graphs. In fact, we show that the problem admits a polynomial 
kernel. Recall that a kernel for a parameterized 
problem $\Qq$ is a polynomial-time algorithm that 
computes for each instance
$(I,k)$ of $\Qq$ an equivalent instance $(I',k')$ with
$|I'|+k'\leq f(k)$ for some computable function $f$. 
The kernel is polynomial if the function $f$ is polynomial. 
We prove that for every instance
$(G,S,T,k)$ of \textsc{CDS-R}, with $G$ planar,  
we can compute in 
polynomial time an instance $(G',S,T,k)$ where
$|V(G')|\leq p(k)$ for some polynomial $p$, $G'$ planar, and
where there exists a reconfiguration sequence under 
\textsf{TAR} from $S$
to $T$ in $G$ (using at most $k$ tokens) if and only if such a sequence exists
in $G'$. 

Our approach is as follows. We first compute a 
small \emph{domination core} for $G$, that is, 
a set of vertices that 
captures exactly the domination properties of $G$ for 
dominating sets of sizes not larger than $k$. While 
the classification of interactions with the domination core
would suffice to solve 
\textsc{Dominating Set Reconfiguration}, additional difficulties arise
for the connected variant. In a second step we use 
planarity to identify large subgraphs that have
very simple interactions with the domination core 
and prove that they can be replaced by constant
size gadgets such that the reconfiguration properties
of $G$ are preserved. 

\subsection{Domination cores}

\begin{definition}
Let $G$ be a graph and let $k\geq 1$ be an integer. A \emph{$k$-domination core}
is a subset $C\subseteq V(G)$ of vertices such that 
every set $X\subseteq V(G)$ of size at most $k$ 
that dominates~$C$ also dominates $G$. 
\end{definition}

It is not difficult to see that \textsc{Dominating Set}
is fixed-parameter tractable on all graphs that admit
a $k$-domination core of size at most $f(k)$ that is 
computable in time $g(k)\cdot n^c$, for any computable
functions $f,g$ and constant $c$. This approach 
was first used (implicitly) in~\cite{DawarK09} to solve
\textsc{Distance-$r$ Dominating Set} on nowhere
dense graph classes. In case $k$ is the size of a 
minimum (distance-$r$) dominating set, one can 
establish the existence of a linear size $k$-domination
core on classes of bounded 
expansion~\cite{DrangeDFKLPPRVS16} (including 
the class of planar graphs) and a polynomial size
(in fact an almost linear size)
$k$-domination core on nowhere dense graph classes~\cite{KreutzerRS17,EickmeyerGKKPRS17}. If $k$ is not
minimum, there exist classes of bounded expansion such 
that a $k$-domination core must have at least quadratic
size~\cite{EibenKMPS18}. The most
general graph classes that admit $k$-domination cores
are given in~\cite{FabianskiPST19}. Moreover, 
\textsc{Dominating Set Reconfiguration} and \textsc{Distance-$r$ 
Dominating Set Reconfiguration} are fixed-parameter tractable on all
graphs that admit small (distance-$r$) $k$-domination
cores~\cite{LokshtanovMPRS18,Siebertz18}.

\begin{lemma}\label{lem:core}
There exists a polynomial
$p$ such that for all $k\geq 1$, every planar graph $G$
admits a polynomial-time computable 
$k$-domination core of size at most $p(k)$. 
\end{lemma}

The lemma is implied by Theorem 1.6 of~\cite{KreutzerRS17}
by the fact that planar graphs are nowhere dense. We want to
stress again that the polynomial size of the $k$-domination 
core results from the fact that $k$ may not be the size
of a minimum dominating set, if $k$ is minimum we can find
a linear size core. Explicit bounds on the degree
of the polynomial can be derived from~\cite{NadaraPRRS18,pilipczuk2018number}, but we refrain from
doing so to not disturb the flow of ideas. 

\smallskip
The following lemma is immediate from the definition of 
a $k$-domination core. 

\begin{lemma}\label{lem:non-adjecent-to-core-irrelevant}
If $D$ is a dominating set of size at most $k$ that contains
a vertex set $W \subset D$ such that $N[D]\cap C= N[D \setminus W]\cap C=C$, 
then $D \setminus W$ is also a dominating set. 
\end{lemma}

\begin{definition}
Let $G$ be a graph and let $A\subseteq V(G)$. The \emph{projection} 
of a vertex $v\in V(G)\setminus A$ into $A$ is the set $N(v)\cap A$. 
If two vertices $u,v$ have the same projection into $A$ we write
$u\sim_A v$. 
\end{definition}

Obviously, the relation $\sim_A$ is an equivalence relation. 
The following lemma is folklore, one possible reference 
is~\cite{GajarskyHOORRVS17}. 

\begin{lemma}\label{lem:neighbourhood-complexity}
Let $G$ be a planar graph and let $A\subseteq V(G)$. 
Then there exists a constant $c$ such that there are 
at most $c\cdot |A|$ different projections to $A$, that is, 
the equivalence relation~$\sim_A$ has at most $c\cdot |A|$
equivalence classes. 
\end{lemma}

\subsection{Reduction rules}
Let $G$ be an embedded planar graph. We say that
a vertex $v$ \emph{touches} a face $f$ if $v$ is
drawn inside $f$ or belongs to the boundary of $f$
or is adjacent to a vertex on the boundary of $f$. 
We fix two connected
dominating sets $S$ and $T$ of size at most~$k$. 
We will present a sequence of lemmas, each of which 
implies a polynomial-time computable reduction rule
that allows to transform~$G$ to a planar graph $G'$ 
that inherits its embedding from $G$, with 
$S,T\subseteq V(G')$ and that has the same 
reconfiguration properties with respect to $S$ and $T$ 
as $G$. To not overload notation, after stating a lemma
with a reduction rule, we assume that the reduction rule
is applied until this is no longer possible and call the
resulting graph again $G$. We also assume that 
whenever one or more of our reduction rules are applicable, 
then they are applied in the order presented. We will 
guarantee that $S$ and $T$ will always 
be connected dominating sets of size at most $k$, hence,
after each application of a reduction rule, we can 
recompute a $k$-domination core in polynomial time. 
This yields only polynomial overhead and allows us to 
assume that we always have marked a $k$-domination core $C$
of size at most $p\coloneqq p(k)$ as described in 
Lemma~\ref{lem:core}. This allows us to state the lemmas
as if $G$ and $C$ were fixed. 
Without loss of generality we assume that~$C$ contains 
$S$ and $T$. 

\begin{definition}
A set $W\subseteq V(G)\setminus C$ of vertices is \emph{irrelevant}
if there is a reconfiguration sequence from $S$ to $T$
in $G$ if and only if there is a reconfiguration sequence
from $S$ to $T$ in $G-W$. 
\end{definition}

\begin{definition}
Let $u,v\in V(G)$ be non-equal vertices. We call the set
$D(u,v)\coloneqq (N(u)\cap N(v))\cup \{u,v\}$ 
the \emph{diamond} induced by $u$ and $v$. 
We call $|N(u)\cap N(v)|$ the \emph{thickness} of 
$D(u,v)$. 
\end{definition}

\begin{lemma}\label{lem:diamond-ends-are-forced}
If $G$ contains a diamond $D(u,v)$ of thickness greater
than $3k$, then at least one of $u$ or $v$ must be 
pebbled in every reconfiguration sequence from $S$ to $T$. 
\end{lemma}

\begin{proof}
Assume $S=S_1,\ldots, S_t=T$ is a reconfiguration sequence
from $S$ to $T$ and $u,v\not\in S_i$ for some $1\leq i\leq t$. 
Then every $s\in S_i$ can dominate
at most $3$ vertices of $N(u)\cap N(v)$: otherwise
$u,v,s$ together with $3$ vertices of $N(u)\cap N(v)$
different from $u,v$ and $s$ would form a complete
bipartite graph $K_{3,3}$. 
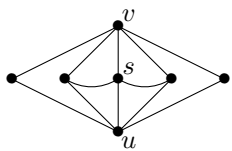
\begin{figure}[h!]
\begin{center}
\begin{tikzpicture}[scale=0.7]
\node (a1) at (0,0) {$\bullet$};
\node at (0.2,-0.2) {$u$};
\node (a2) at (0,2) {$\bullet$};
\node at (0.2,2.2) {$v$};
\node at (0.2,1.2) {$s$};
\draw (0,1) to[bend right=30] (1,1);
\draw (-1,1) to[bend right=30] (0,1);
\foreach \x in {-2cm,-1cm,0cm,1cm,2cm}
{
	\node at (\x,1) {$\bullet$};
	\draw[-] (0,2) -- (\x,1) -- (0,0);
}
\end{tikzpicture}
\end{center}
\caption{A vertex $s\in S_i$ can dominate at most $3$ vertices
of $N(u)\cap N(v)$.}
\end{figure}
\end{proof}

\begin{lemma}\label{lem:reduce-diamond-edges}
If $G$ contains a diamond $D(u,v)$ of thickness greater
than $3k$, then we can remove all internal edges in $D(u,v)$, i.e., edges with both endpoints in $N(u) \cap N(v)$. 
\end{lemma}
\begin{proof}
Assume $S=S_1,\ldots, S_t=T$ is a reconfiguration sequence
from $S$ to $T$. 
According to Lemma~\ref{lem:diamond-ends-are-forced}, for 
each $1\leq i\leq t$, $S_i\cap \{u,v\}\neq \emptyset$. 
Hence all vertices of $N(u) \cap N(v)$ are always
dominated by at least one of $u$ or $v$, say by $u$. Moreover, pebbling more than one vertex of $N(u) \cap N(v)$ will never create
connectivity via internal edges that is not already there via edges incident on $u$. 
In other words, for any connected dominating set $S$ of $G$, if an edge $yz$ is used for connectivity, where $y,z \in N(u) \cap N(v)$, 
then this edge can be replaced by either the path $yuz$ or the path $yvz$ (depending on which of $u$ or $v$ is in $S$).
\end{proof}

As described earlier, we now apply the reduction rule of 
Lemma~\ref{lem:reduce-diamond-edges} until this is no 
longer possible, and name the resulting graph again $G$. 
As we did not make use of the properties of a $k$-domination
core in the lemma, it is sufficient to recompute a $k$-domination 
core $C$ after applying the reduction rule exhaustively. In the 
following it may be necessary to recompute it after each application
of a reduction rule. We will not mention these
steps explicitly anymore in the following.  

\begin{lemma}[$\star$]\label{lem:reduction}
If $G$ contains a diamond $D(u,v)$ of thickness greater
than $4|C| + 3k + 1$ then~$G$ contains an irrelevant vertex. 
\end{lemma}

We may in the following assume 
that $G$ does not contain diamonds of thickness 
greater than $4|C| + 3k + 1$. 

\begin{corollary}\label{lem:high-degree-left}
If a vertex $v \in V(G)$ has degree greater 
than $(4|C| + 3k + 1)\cdot k$, 
then the token on $v$ is never lifted throughout
a reconfiguration sequence. 
\end{corollary}

\begin{proof}
Assume $S=S_1,\ldots, S_t=T$ is a reconfiguration sequence 
from $S$ to $T$ in $G$ and assume there is 
$S_i$ with $v\not\in S_i$. The dominating set $S_i$ has 
at most $k$ vertices and must dominate $N(v)$. Hence, 
there must be one vertex $u\in S_i$ that dominates at 
least a $1/k$ fraction of $N(v)$, which is larger than 
$4|C| + 3k + 1$. Then there is a diamond
$D(u,v)$ of thickness greater than $4|C| + 3k + 1$, which does 
not exist after application of the reduction rule of Lemma~\ref{lem:reduction}. 
\end{proof}

According to Corollary~\ref{lem:high-degree-left}, the only
vertices that can have high degree after applying the
reduction rules are vertices that 
are never lifted throughout a reconfiguration 
sequence. 
This gives rise to another reduction rule that is 
similar to the rule of Lemma~\ref{lem:reduce-diamond-edges}.

\begin{lemma}\label{lem:reduce-long-paths}
Assume $v$ is a vertex of degree greater than 
$(4|C| + 3k + 1)\cdot k$. Then we may
remove all edges with both endpoints in $N(v)$. 
\end{lemma}

\begin{proof}
Let $G'$ be the graph 
obtained from $G$ by removing all edges with both
endpoints in~$N(v)$. We claim that reconfiguration between 
$S$ and $T$ is possible in $G$ if and only if it is possible
in $G'$. The fact that $S$ and $T$ 
are in fact connected dominating sets in $G'$ 
is implied by the argument below. 

Assume $S=S_1,\ldots, S_t=T$ is a reconfiguration sequence 
from $S$ to $T$ in $G$. We claim that the same sequence
is a reconfiguration sequence in $G'$. 
According to Corollary~\ref{lem:high-degree-left}, 
$v\in S_i$ for all $1\leq i\leq t$. This implies that 
$S_i$ is connected in $G'$ for all $1\leq i\leq t$,
as all $x,y\in S_i$ that are no longer connected by 
an edge in $G'$ but were connected in $G$ 
are connected via a path of length $2$ using the vertex $v$. 
It is also easy to see that $S_i$ is a dominating set in $G'$, 
as all vertices that are no longer dominated by $s\in S_i$ in $G$
are still dominated by $v$. 
Observe that this in particular implies that $S$ and $T$ 
are connected dominating sets in $G'$. 
Vice versa, if $S=S_1,\ldots, S_t=T$ is a reconfiguration sequence 
from $S$ to $T$ in $G'$, this is trivially also a reconfiguration
sequence in $G$. 
\end{proof}

The following reduction rule is obvious. 
\begin{lemma}\label{lem:remove-isolated-vertices}
If a vertex $v$ has more than $k + 1$ pendant neighbours, i.e., neighbors of degree exactly one, 
then it suffices to retain exactly $k + 1$ of them in the graph.  
\end{lemma}

\begin{lemma}\label{lem:projection2-small}
There are at most $c|C| \cdot (4|C|+3k+1)$ vertices of 
$V(G)\setminus C$ that have $2$ neighbours in $C$, 
where $c$ is the constant of Lemma~\ref{lem:neighbourhood-complexity}. 
\end{lemma}

\begin{proof}
According to Lemma~\ref{lem:neighbourhood-complexity}
there are at most $c|C|$ different projections to $C$. Each 
projection class that has at least $3$ representatives
has size at most $2$, as otherwise we would find
a $K_{3,3}$ as a subgraph, contradicting the planarity of $G$. 
Consider a class with a projection of size
$2$ into $C$. Denote these two vertices of $C$ by $u$ and $v$. 
If this class has more than $4|C|+3k+1$ representatives, then $D(u,v)$ is a
diamond of thickness greater than $4|C|+3k+1$, which cannot 
not exist after exhaustive application of the reduction rule of 
Lemma~\ref{lem:reduction}. 
\end{proof}

We now come to the description of our final reduction rule. 
Let $D$ denote the set of vertices containing both $C$ and all vertices of $V(G) \setminus C$ having at least two neighbors in $C$. 
In other words, $D$ contains all those vertices in $V(G) \setminus C$ that have exactly one neighbor in~$C$. 
According to Lemma~\ref{lem:projection2-small} at most 
$c|C|\cdot (4|C|+3k+1)$ vertices have two neighbors in~$C$, hence
$|D|\leq c|C|\cdot (4|C|+3k+1)+|C|\eqqcolon p$. 

\begin{lemma}[$\star$]\label{lem:reduction-paths}
Assume there are two vertices $u$ and $v$ with degree greater
than $4p+ (4|C| + 3k + 1)\cdot k + 1$. Let 
$\Pp$ be a maximum set of vertex-disjoint paths
of length at least $2$ that run between $u$ and $v$ 
using only vertices in $V(G) \setminus D$. 
If $|\mathcal{P}| > 4p + (4|C| + 3k + 1)\cdot k + 1$, then there 
is $G'$ such that the instances $(G,S,T,k)$ and $(G',S,T,k)$ 
are equivalent, $G'$ is planar, and $|V(G')|< |V(G)|$.
\end{lemma}

We are ready to state the final result.  

\begin{theorem}
\textsc{CDS-R} under \textsf{TAR} parameterized by
$k$ admits a polynomial kernel on planar graphs. 
\end{theorem}

\begin{proof}
Our kernelization algorithm starts by computing (in polynomial time) a $k$-domination core $C$ of size at most 
$p \coloneqq p(k)$ as described in 
Lemma~\ref{lem:core}. Without loss of generality
we assume that $C$ contains $S$ and $T$. After each application 
of a reduction rule, we recompute the core, giving a polynomial 
blow-up of the running time. We are left to prove that each 
reduction rule can be implemented in polynomial time and that 
we end up with a polynomial number of vertices. 

It is clear that the reduction rules of Lemma~\ref{lem:reduction}, 
Lemma~\ref{lem:reduce-long-paths} and Lemma~\ref{lem:remove-isolated-vertices} can easily be implemented in polynomial time.
The reduction rule of Lemma~\ref{lem:reduction-paths} is slightly
more involved, however, we can use a standard maximum-flow 
algorithm to compute in polynomial time a maximum 
set of vertex-disjoint paths in a subgraph of $G$. 

It remains to bound the size of $G$. Recall that we call 
$D$ the set of all vertices $C$ and of all vertices 
of $V(G)\setminus C$ that have
at least $2$ neighbors in $C$. It follows from 
Lemma~\ref{lem:projection2-small} that~$D$ has size 
at most $c|C| \cdot (4|C|+3k+1)+ |C| \eqqcolon p$, 
where $c$ is the constant of 
Lemma~\ref{lem:neighbourhood-complexity}. 
We are left to bound the number of vertices in 
$V(G)\setminus C$ having exactly one neighbour in $C$ 
(recall that each vertex in $V(G) \setminus C$ has at 
least one neighbour in $S \cup T \subseteq C$). 

Let $p' = (4p+(4|C|+3k+1)\cdot k+1)\cdot 
(4|C|+3k+1)\cdot k+k+1$, which is still 
a polynomial in $k$. Towards a contradiction, 
assume that there exists an equivalence 
class $Q$ in $\sim_C$ with a projection of size one 
containing more than $p'$ vertices.  
Let $u\in C$ denote the projection of the aforementioned class. 
Due to Lemma~\ref{lem:remove-isolated-vertices}, we know that at most $k + 1$ of the vertices in~$Q$ are pendant, i.e., adjacent to only $u$ in $G$. 
Since we cannot apply the reduction rule of 
Lemma~\ref{lem:reduce-long-paths} any more, we know 
that there are no edges with both endpoints in $Q$. 
Hence, all but $k + 1$ vertices of $Q$ must be 
adjacent to at least one other vertex in $V(G) \setminus C$. 
Let $R = N_G(Q) \setminus \{u\}$ denote this set of neighbours. 
No vertex in $R$ can be adjacent to more than $4|C|+3k+1$
vertices of $Q$, as we cannot apply the reduction rule of 
Lemma~\ref{lem:reduction}. The vertices of $R$ must be 
dominated by $S$, and cannot be dominated by $u$, as
otherwise two neighbours of $u$ would be connected. Hence, 
there is $v\in S$ different from $u$ that dominates at least
a $1/k$ fraction of $R$. This implies the existence of at 
least $4p + (4|C| + 3k + 1)\cdot k + 1$ vertex-disjoint paths
of length at least $2$ that run between $u$ and $v$. 
But in this case, the reduction rule of 
Lemma~\ref{lem:reduction-paths} is applicable. Therefore, we conclude that $Q$ cannot exist, obtaining 
a bound on the size of all equivalence classes of $\sim_C$, as needed. 
\end{proof}
\bibliography{ref}

\newpage
\appendix
\section{Details omitted from Section 3}

We need the following lemma to prove Lemma~\ref{lem:forward}.  

\begin{lemma}\label{lem:reconf-tree}
Let $T_1,T_2$ be two trees on vertex set $\{1,\ldots, k\}$ and let 
$f_1,\ldots f_{k-1}$ be an ordering of the edges in $T_2$. Then, in polynomial time, we can find 
an ordering $e_1,\ldots,e_{k-1}$ of the edges in $T_1$ such that the following holds. 
In the sequence of graphs $T_0',T_1',\ldots,T_{k-1}'$ on vertex set $\{1,\ldots, k\}$, where  
for each $0\leq i\leq k-2$, $T_{i+1}'=T_i'+f_i-e_i$ and $T_0'=T_1$, 
we have that $T_{i}'$ is a tree, for all $i\in [k-1]$, and $T_{k-1}'=T_2$.   
\end{lemma}

\begin{proof}
We proceed by induction on $\ell=\vert E(T_1)\setminus E(T_2)\vert$. 
In the base case, we have $\ell=0$ and $E(T_1)= E(T_2)$. 
In this case $f_1,\ldots f_{k-1}$ is also the required ordering of the edges in $T_1$ 
(note that the sequence of graphs consists of only $T_1 = T_2$ in this case). 

Now consider the induction step, $\ell>1$. 
Let $j$ be the first index in $\{1,\ldots,k-1\}$ such that $f_j\notin E(T_1)$. 
We add $f_j$ to $T_1$ and this creates a cycle in $T_1$. Hence, there exists an edge $e_j\in E(T_1)\setminus E(T_2)$ whose removal results in a tree. 
That is, $T_1'=T_1+f_j-e_j$ is a tree. Notice that $\vert E(T_1')\setminus E(T_2)\vert =\ell-1$. 
By the induction hypothesis, there is a sequence $g_1,\ldots,g_{k-1}$ of edges in $E(T_1')$ such that 
for the sequence of graphs 
$T_1'=T_0'',T_1'',\ldots,T_{k-1}''$ on vertex set $\{1,\ldots, k\}$,  we have
$T_{i+1}''=T_i''+f_i-g_i$,  each $T_i''$ is a tree, and $T_2=T_{k-1}''$, $0\leq i<k$.    
Since $j$ is the first index in $\{1,\ldots,k-1\}$ such that $f_j\notin E(T_1)$, $T_1'=T_1+f_j-e_j$,  and $T_0'',T_1'',\ldots,T_{k-1}''$ are trees, we have that $g_i=f_i$ for 
all $i < j$. Notice that $f_j\in E(T_1')$ and $E(T_1)=(E(T_1')\setminus \{f_j\})\cup \{e_j\}$. 

 We claim that $e_1,\ldots,e_{j-1},e_j,e_{j+1},\ldots,e_{k-1}$, where $e_i=g_i$ for all $i < j$,   
is the required sequence of edges in $T_1$. 
Let $T_0',T_1',\ldots,T_{k-1}'$ be the sequence  where,  
for each $0\leq i\leq k-2$, $T_{i+1}'=T_i'+f_i-e_i$ and $T_0'=T_1$. 
Since $g_i=f_i=e_i$  for all $i< j$, we have that $T_1=T_0'=T_1'=\ldots=T_{j-1}'$. 
Moreover, $T_j'=T_1+\{f_1,\ldots,f_j\}-\{e_1,\ldots,e_j\}=T_1+\{f_1,\ldots,f_j\}-\{g_1,\ldots,g_j\}=T_j''$ 
because $E(T_1)=(E(T_1')\setminus \{f_j\})\cup \{e_j\}$ and $e_i=g_i$ for all $i < j$. 
Then, the sequence $T_j',\ldots,T'_{k-1}$ is the same as the sequence $T_j'',\ldots,T''_{k-1}$. 
Therefore, the sequence $e_1,\ldots,e_{j-1},e_j,e_{j+1},\ldots,e_{k-1}$  of edges in $T_1$ satisfies the conditions of the lemma.  
%
%
\end{proof}

\subsection*{Proof of Lemma~\ref{lem:fourdeg}}

\begin{proof}
We iteratively remove minimum degree vertices and show that
we can always remove a vertex of degree at most $4$ in each step. 
\begin{itemize}
\item Every subdivision vertex $w \in W(H^i)$ for $1\leq i\leq k$
has degree at most $4$; it has $4$ neighbors in $V(H^i)\cup V(H^{i+1})$. 
\item After removal of all subdivision vertices the degree of the remaining
vertices of each $H^i$ is at most one. That is, a vertex in $H^{(1,1)}$ may have a neighbor in $\{w_2,\ldots,w_{k}\}$.
\item After the removal of $V(H^1)\cup\ldots V(H^k)$, the degree of all vertices except $v_1$ and $x_k$
is at most $2$. 
\item Finally we remove $v_1$ and $x_k$.
\qedhere
\end{itemize}
\end{proof}

\subsection*{Proof of Lemma~\ref{lem:forward}}

\begin{proof}
We aim
to shift the connected vertices of $Q_s$ through the subgraphs
$H^1,\ldots, H^k$ (in that order) to maintain connectivity and eventually shift to $Q_t$.  
For each $u_i\in V(G)$, $1\leq j\leq k$ and $1\leq r\leq 20k$, we use  $u_i^{(j,r)}$ to denote the copy of $u_i$ in $H^{(j,r)}$. 

Let $C=\{u_1,\ldots,u_k\}$ be a $k$-colored clique in $G$ such that $c(u_i)=i$, for all $1\leq i\leq k$. 
To prove the lemma, we need to define a reconfiguration sequence starting from $Q_s$ and ending at $Q_t$ 
such that the cardinality of any solution in the sequence is at most $2k$. 
First we define $k$ ``colored'' trees $\widehat{T}_1,\ldots, \widehat{T}_k$ each on $2k-1$ vertices, and then 
prove that there are reconfiguration sequences from $Q_s$ to $V(\widehat{T}_1)$,  $V(\widehat{T}_i)$ 
to $V(\widehat{T}_{i+1})$ for all $1\leq i<k$, and $V(\widehat{T}_k)$ to $Q_t$. 

We start by defining $\widehat{T}_1,\ldots, \widehat{T}_k$. 
For each $1\leq i\leq k$, $C_i=\{u^{(i,1)}_1,\ldots,u_k^{(i,1)}\}$ and $S_i=\{z\in V(H^{(i,1)})\colon N_{H^{(i,1)}}(z)\cap C_i=2\}$. 
That is, for each $1\leq j\leq k$ and $j\neq i$,  $s_{u_i^{(i,1)}u_j^{(i,1)}}\in S_i$ (the subdivision vertex on the edge $u_i^{(i,1)}u_j^{(i,1)}$ is in $S_i$), 
and $\vert S_i\vert =k-1$. In other words, $C_i$ contains the copies of the vertices of the clique $C$ in $H^{(i,1)}$ and $S_i$ contains subdivision vertices corresponding to $k-1$ edges in the clique incident on the $i$th colored vertex of the clique, such that $H[C_i\cup S_i]$ is a tree. Now, define $\widehat{T}_i=H[C_i\cup S_i]$. 
It is easy to verify that $\widehat{c}(C_i\cup S_i)=\{1,\ldots,k+1\}$ and hence $C_i\cup S_i=V(\widehat{T}_i)$ 
is a solution to the \textsc{CCS} instance $(H,\widehat{c},2k)$. 
Let $T_s=H[Q_s]$ and $T_t=H[Q_t]$. Note that $T_s$ and $T_t$ are trees on $2k-1$ vertices each. 

\medskip
\noindent{\bf Case 1:  Reconfiguration from $Q_s$ to $V(\widehat{T}_1)$.} 
Informally, we move to $\widehat{T}_1$ by adding a token on $u_i^{(1,1)}$ and then removing tokens from $v_i$ for $i$ in the order  
$2,\ldots,k,1$ (for a total of $2k$ token additions/removals). Finally, we move the tokens from $\{w_2,\ldots,w_{k-1}\}$ to $S_1$ in $2(k-1)$ steps.  
The length of the reconfiguration sequence is $2k+2(k-1)=4k-2$. 

Formally, we define $Z_0=Q_s$ and for each $1\leq j\leq k-1$, $Z_{2j-1}=Z_{2j-2}\cup \{u_{j+1}^{(1,1)}\}$ and $Z_{2j}=Z_{2j-1}\setminus \{v_{j+1}\}$. That is, for each $1\leq j\leq k-1$, 
\begin{eqnarray*}
Z_{2j-1}&=&\{u_2^{(1,1)},\ldots,u^{(1,1)}_{j+1}\}\cup \{v_{j+1}\ldots,v_k,v_1\}\cup\{w_1,\ldots,w_{k-1}\}, \mbox{ and }\\
Z_{2j}&=&\{u_2^{(1,1)},\ldots,u^{(1,1)}_{j+1}\}\cup \{v_{j+2}\ldots,v_k,v_1\}\cup\{w_1,\ldots,w_{k-1}\}.
\end{eqnarray*}
Next, we define $Z_{2k-1}$ and $Z_{2k}$ as 
\begin{eqnarray*}
Z_{2k-1}&=&\{u_2^{(1,1)},\ldots,u^{(1,1)}_k,u_{1}^{(1,1)}\}\cup \{v_1\}\cup\{w_1,\ldots,w_{k-1}\}, \mbox{ and }\\
Z_{2k}&=&\{u_2^{(1,1)},\ldots,u^{(1,1)}_k,u_{1}^{(1,1)}\}\cup\{w_1,\ldots,w_{k-1}\}.
\end{eqnarray*}

It is easy to verify that $Z_1,\ldots Z_{2k}$ are solutions to the \textsc{CCS} instance $(H,\widehat{c},2k)$. 
Thus, we now have a reconfiguration sequence $Z_0,Z_1,\ldots,Z_{2k}$, where $Z_0=Q_s$. 

Next, we explain how to get a reconfiguration sequence from $Z_{2k}$ to $V(\widehat{T}_1)$. 
Recall that $Z_{2k}=C_1\cup \{w_1,\ldots,w_{k-1}\}$ and $V(\widehat{T}_1)=C_1\cup S_1$. 
Let $s_j=s_{u_1^{(1,1)}u_j^{(1,1)}}$, for all $2\leq j\leq k$. 
Notice that $S_1=\{s_2,\ldots,s_k\}$. To obtain a reconfiguration sequence from $Z_{2k}$ to $V(\widehat{T}_1)$, we add $s_j$ 
and then remove $w_j$ for $j$ in the order $2,\ldots,k$. 
Since $w_j$ and $s_j$ connect the same two vertices from $C_1$, this reconfiguration sequence will maintain connectivity. 
Moreover, it is easy to verify that each set in the reconfiguration sequence 
uses all the colors $\{1,\ldots,k\}$. Therefore, there exists a reconfiguration sequence of length $4k-2$ from $Q_s$ to $V(\widehat{T}_1)$. 

\medskip
\noindent{\bf Case 2:  Reconfiguration from $V(\widehat{T}_i)$ to $V(\widehat{T}_{i+1})$.} First we define $20k$ trees $P_1,\ldots P_{20k}$, each on $2k-1$ vertices 
such that for all $1\leq r\leq 20k$, $(i)$ $V(P_r)\subseteq V(H^{(i,r)})$, and $(ii)$ $\widehat{T}_i=P_1$. 
Then we give a reconfiguration sequence from $V(P_r)$ to $V(P_{r+1})$ for all $r\in [20k-1]$ and a reconfiguration sequence from $V(P_{20k})$ to $V(\widehat{T}_{i+1})$.  

Recall that $C=\{u_1,\ldots,u_k\}$ is a $k$-colored clique in $G$ such that $c(u_i)=i$ for all $1\leq i\leq k$. 
For each $1\leq r\leq 20k$, let $C_i^r=\{u^{(i,r)}_1,\ldots,u_k^{(i,r)}\}$ and $S_i^r=\{z\in V(H^{(i,r)})\colon N_{H^{(i,r)}}(z)\cap C_i^r=2\}$. 
That is, for each $1\leq j\leq k$ and $j\neq i$,  $s_{u_i^{(i,r)}u_j^{(i,r)}}\in S_i^r$ 
(i.e, the subdivision vertex on the edge $u_i^{(i,r)}u_j^{(i,r)}$ is in $S_i^r$) and $\vert S_i^r\vert =k-1$. 
Let $P_r=H[C_i^r\cup S^r_i]$. 
Notice that  for all $r\in [20k]$, $P_r$
is a tree on $2k-1$ vertices.   
Moreover, for each $1\leq r\leq 20k$, $V(P_r)$ is a solution to the \textsc{CCS} instance $(H,\widehat{c},2k)$.

\medskip
\noindent{\bf Case 2(a):  Reconfiguration from $V(P_r)$ to $V(P_{r+1})$.} 
By arguments similar to those given for Case $1$, one can prove that  there is a reconfiguration sequence of length  
$4k-2$ from $V(P_r)$ to $V(P_{r+1})$, for all $1\leq r<20k$. For completeness we give the details here.  
Fix an integer $1\leq r<20k$. Let $s_j=s_{u_i^{i,r}u_j^{(i,r)}}$
and $s'_j=s_{u_i^{(i,r+1)}u_j^{(i,r+1)}}$ for all $j\in \{1,\ldots,k\}\setminus \{i\}$. 
Notice that $S_i^r=\{s_j \colon j\in \{1,\ldots,k\}\setminus \{i\}\}$ and $S_i^{r+1}=\{s'_j \colon j\in \{1,\ldots,k\}\setminus \{i\}\}$. 
Now we define $Z_0=V(P_r)=C_i^r\cup S_i^r$ and for each $1\leq j\leq i-1$, $Z_{2j-1}=Z_{2j-2}\cup \{u_j^{(i,r+1)}\}$ and $Z_{2j}=Z_{2j-1}\setminus \{u_j^{(1,r)}\}$. That is, for each $1\leq j\leq i-1$, 
\begin{eqnarray*}
Z_{2j-1}&=&\{u_1^{(i,r+1)},\ldots,u^{(i,r+1)}_j\}\cup \{u_j^{(i,r)}\ldots,u_k^{(i,r)}\}\cup S_i^r, \mbox{ and }\\
Z_{2j}&=&\{u_1^{(1,1)},\ldots,u^{(1,1)}_j\}\cup \{u_{j+1}^{(i,r)}\ldots,u_k^{(i,r)}\}\cup S_i^r.
\end{eqnarray*}

For each $i\leq j\leq k-1$, $Z_{2j-1}=Z_{2j-2}\cup \{u_{j+1}^{(i,r+1)}\}$ and $Z_{2j}=Z_{2j-1}\setminus \{u_{j+1}^{(1,r)}\}$. That is, for each $i\leq j\leq k-1$, 
\begin{eqnarray*}
Z_{2j-1}&=&\{u_1^{(i,r+1)},\ldots,u^{(i,r+1)}_{i-1},u^{(i,r+1)}_{i+1},\ldots, u^{(i,r+1)}_{j+1}\}\cup \{u_{j+1}^{(i,r)}\ldots,u_k^{(i,r)},u_i^{(i,r)}\}\cup S_i^r, \mbox{ and }\\
Z_{2j}&=&\{u_1^{(i,r+1)},\ldots,u^{(i,r+1)}_{i-1},u^{(i,r+1)}_{i+1},\ldots, u^{(i,r+1)}_{j+1}\}\cup \{u_{j+2}^{(i,r)}\ldots,u_k^{(i,r)},u_i^{(i,r)}\}\cup S_i^r.
\end{eqnarray*}
Next, we define $Z_{2k-1}$ and $Z_{2k}$ as 
\begin{eqnarray*}
Z_{2k-1}&=&\{u_1^{(i,r+1)},\ldots,u^{(i,r+1)}_{k}\} \cup \{u_i^{(i,r)}\}\cup S_i^r, \mbox{ and }\\
Z_{2k}&=&\{u_1^{(i,r+1)},\ldots,u^{(i,r+1)}_{k}\} \cup S_i^r.
\end{eqnarray*}

Next, for each $1\leq j\leq k-1$, let $Z_{2k+2j-1}=Z_{2k+2j-2}\cup \{s_j'\}$ and $Z_{2k+2j}=Z_{2k+2j-1}\setminus \{s_j\}$. 
It is easy to verify that $Z_1,\ldots Z_{4k-2}$ are solutions to the \textsc{CCS} instance $(H,\widehat{c},2k)$ 
and $Z_0,\ldots,Z_{4k-2}$ is a reconfiguration sequence  where $Z_0=V(P_r)$ and $Z_{4k-2}=V(P_{r+1})$. 

\medskip
\noindent{\bf Case 2(b):  Reconfiguration from $V(P_{20k})$ to $V(\widehat{T}_{i+1})$.} 
Next, we explain how to get a reconfiguration sequence from $V(P_{20k})$ to $V(\widehat{T}_{i+1})$ using Lemma~\ref{lem:reconf-tree}. 
Recall that $C_i^{20k}=\{u^{(i,20k)}_1,\ldots,u_k^{(i,20k)}\}$ and $S_i^{20k}=\{z\in V(H^{(i,20k)})\colon N_{H^{(i,20k)}}(z)\cap C_i^{20k}=2\}$. 
Let $C_{i+1}=\{u^{(i+1,1}_1,\ldots,u_k^{(i+1,1)}\}$ and $S_{i+1}=\{z\in V(H^{(i+1,1)})\colon N_{H^{(i+1,1)}}(z)\cap C_{i+1}=2\}$. 
For ease of presentation, let $s_j=s_{u_i^{(i,20k)}u_{j}^{(i,20k)}}$ for all $j\in \{1,\ldots,k\}\setminus \{i\}$. 
Also, let $s_j'=s_{u_i^{(i+1,1)}u_{j}^{(i+1,1)}}$ for all $j\in \{1,\ldots,k\}\setminus \{i+1\}$. That is, $S_i^{20k}=\{s_j \colon j\in \{1,\ldots,k\}\setminus \{i\}\}$ and 
$S_{i+1}=\{s'_j \colon j\in \{1,\ldots,k\}\setminus \{i+1\}\}$. Notice that $V(P_{20k})=C_i^{20k}\cup S_{i}^{20k}$ and  $V(\widehat{T}_{i+1})=C_{i+1}\cup S_{i+1}$. 

Towards proving the required reconfiguration sequence, we give a reconfiguration sequence from $C_i^{20k}\cup S_{i}^{20k}$ 
to $C_{i+1}\cup S_{i}^{20k}$  and then from $C_{i+1}\cup S_{i}^{20k}$ to $C_{i+1}\cup S_{i+1}$. 
The reconfiguration sequence from  $C_i^{20k}\cup S_{i}^{20k}$ 
to $C_{i+1}\cup S_{i}^{20k}$ is similar to the one in Case 1. That is, we add $u_j^{(i+1,1)}$ and delete $u_j^{(i,20k)}$ 
for $j$ in the order $1,\ldots,i-1,i+1,\ldots,k,i$. This gives a reconfiguration sequence from $C_i^{20k}\cup S_{i}^{20k}$ to $Z=C_{i+1}\cup S_{i}^{20k}$ of length $2k$.

Next we explain how to get a reconfiguration sequence from $Z=C_{i+1}\cup S_{i}^{20k}$ to $C_{i+1}\cup S_{i+1}$. 
Notice that $H[Z]$ and $\widehat{T}_{i+1}=H[C_{i+1}\cup S_{i+1}]$ are trees.   
Recall that $T^i$ is the star on $\{1,\ldots,k\}$ with vertex $i$ being the center, and $T^{i+1}$ is 
is the star on $\{1,\ldots,k\}$ with vertex $i$ being the center. 
Also, $c_j$ is a coloring on $H^j$ which is inherited from the coloring $c$ of $G$. That is, 
$c_{i+1}(u_j^{(i+1,1)})=j$ for all $1\leq j\leq k$.
Then, $H[Z]=K[C_{i+1}]\upharpoonright_{c_{i+1}} T^i$ and $\widehat{T}_{i+1}=H[C_{i+1}\cup S_{i+1}]=K[C_{i+1}]\upharpoonright_{c_{i+1}} T^{i+1}$.

Let $e_1^{i+1},\ldots,e_{k-1}^{i+1}$ be an arbitrary ordering of the the edges in $T^{i+1}$.
By Lemma~\ref{lem:reconf-tree},  we  have a sequence $e_1^i,\ldots,e_{k-1}^i$ of edges in $T^i$ such that 
for the sequence $T^i_0,T^i_1,\ldots,T^i_{k-1}$ on  vertex set $\{1,\ldots, k\}$,  where  
for each $0\leq j\leq k-2$, $T^i_{j+1}=T_j^i+e_j^{i+1}-e_j^i$ and $T^i_0=T^i$, the following holds.
\begin{itemize}
\item[(i)] $T_{j}^i$ is a tree for all $0\leq j\leq k-1$, and 
\item[(ii)] $T^i_{k-1}=T^{i+1}$. 
\end{itemize}

This implies that, from the sequences $e_{1}^i,\ldots,e_{k-1}^i$ and $e_1^{i+1},\ldots,e_{k-1}^{i+1}$, we get a sequence $f_1,\ldots,f_{k-1}'$ on $S_i^{20k}$ and a sequence $f'_1,\ldots,f'_{k-1}$ on $S_{i+1}$ such that 
the for the sequence $L_0,\ldots, L_{2(k-1)}$, where $L_0=C_{i+1}\cup \{f_1,\ldots,f_{k-1}\}$ and  for all $1\leq j\leq k-1$
$L_{2j-1}=(L_{2j-2} \cup \{f'_i\})$, $L_{2j}=L_{2j-1}\setminus \{f_i\}$ 
the following holds. 

\begin{itemize}
\item[(1)] $H[L_{i}]$ is connected for all $0\leq i\leq k-1$,  and 
\item[(2)] $L_{k-1}=S_{i+1}\cup C_{i+1}$.
\end{itemize}

Here, conditions $(1)$ and $(2)$ follow from conditions $(i)$ and $(ii)$, respectively. 
Moreover, $\widehat{c}(L_i)=[k+1]$ for all $0\leq i\leq 2(k-1)$ and $L_0=Z$. 
Thus, $L_0,\ldots, L_{2(k-1)}$ is a valid reconfiguration sequence from $Z$ to $V(\widehat{T}_{i+1})$. 
Note that the ordering on the edges implies an ordering by which we can move the subdivision vertices from 
$S_{i}$ to $S_{i+1}$ without violating connectivity. 
This implies that there is a reconfiguration sequence from $V(P_{20k})$ to $V(\widehat{T}_{i+1})$, of length $4k-2$. 
Therefore, we have a reconfiguration sequence from $V(\widehat{T}_i)$ to $V(\widehat{T}_{i+1})$ of length $\Oof(k^2)$. 

\medskip
\noindent{\bf Case 3:  Reconfiguration from $V(\widehat{T}_k)$ to $V({T}_t)$.} 
The arguments for this case are similar to those given in Case 1, we therefore omit the details.  
By summing up the lengths of reconfiguration sequences, we get that if $(G,c,k)$ is a \yes-instance of  
\textsc{Multicolored Clique} then there is a reconfiguration sequence from $Q_s$ to $Q_t$, of length $\Oof(k^3)$. 
\end{proof}

\subsection*{Proof of Lemma~\ref{lem:backward}}
\begin{proof}
For each $1\leq i\leq k+1$, let $Q_i$ be the set of vertices colored by the color $i$. 
That is, $Q_i=\widehat{c}^{-1}(i)$. First, we prove some auxiliary claims. 
The proofs of the following two claims follow from the construction of $H$ and the definition of $\widehat{c}$. 

\begin{claim}
\label{clm:Qindp}
(i) $Q_1\cup \ldots \cup Q_k$ is an independent set in $H$, and $(ii)$ every  vertex in $Q_{k+1}$ is connected to vertices of at most two distinct colors. 
\end{claim}

\begin{claim}
\label{clm:samecopyvert}
Let $v,w\in V(H)\setminus (V(H^0)\cup V(H^{k+1}))$ be two distinct vertices such that $\widehat{c}(v)=\widehat{c}(w)$ and $\widehat{c}(v)\in \{1,\ldots,k\}$. 
If $v$ and $w$ have a common neighbor in $V(H)\setminus V(H^0)$, 
then $v$ and $w$ are copies of same vertex $z\in V(G)$. 
\end{claim}

\begin{claim}
\label{clm:small}
Let $Y\subseteq V(H)$ be a vertex subset such that $\widehat{c}(Y)=\{1,\ldots,k+1\}$ and $H[Y]$ is connected. Then, $\vert Y\vert \geq 2k-1$.  
\end{claim}

\begin{proof}\renewcommand{\qedsymbol}{\ensuremath{\lrcorner}}
Let $B=Y\setminus \widehat{c}^{-1}(k+1)=Y\cap (Q_1\cup \ldots \cup Q_k)$.   
Since $\widehat{c}(Y)=\{1,\ldots,k+1\}$, $\vert B\vert \geq k$ and by Claim~\ref{clm:Qindp}$(i)$, $B$ is an independent set in $H$. By Claim~\ref{clm:Qindp}$(ii)$, each vertex in $Q_{i+1}$ is connected to vertices of at most two distinct colors. Thus, since $H[Y]$ is connected, the claim follows.  
\end{proof}

Suppose $(H,\widehat{c},Q_s,Q_t,2k)$ is a \yes-instance of  
\textsc{CCS-R}. Then, there is a reconfiguration sequence $D_1,\ldots,D_{\ell}$ for $\ell\in {\mathbb N}$, where 
$D_1=Q_s$ and $D_{\ell}=Q_t$. 
Without loss of generality, we assume that the sequence $D_1,\ldots, D_{\ell}$ is a minimal reconfiguration sequence. 
Then, by Claim~\ref{clm:small}, for each $i\in [\ell]$, $2k-1\leq \vert D_i\vert\leq 2k$. 

Moreover, since $\vert D_1\vert =\vert D_{\ell}\vert=2k-1$, we have that for each even $i$, $D_i$ is obtained from $D_{i-1}$ by a token addition, and for each odd $i$, $D_i$ is obtained from $D_{i-1}$ by a token removal. This also implies that for each even $i$, $\vert D_i\vert =2k$, for each odd $i$, $\vert D_i\vert=2k-1$, and $\ell$ is odd.

\begin{claim}
\label{clm:allcolorsol}
Let $i\in [\ell]$ and $\vert D_i\vert =2k-1$. Then, for all $1\leq j\leq k$, $\vert D_i\cap Q_j\vert=1$, and $\vert D_i\cap Q_{k+1}\vert=k-1$. Moreover, each vertex in  $D_i\cap Q_{k+1}$ will be adjacent to exactly two vertices in $H[D_i]$  and these vertices will be of different colors from $\{1,\ldots,k\}$. 
\end{claim}

\begin{proof}\renewcommand{\qedsymbol}{\ensuremath{\lrcorner}}
By Claim~\ref{clm:Qindp},  $Q_1\cup \ldots \cup Q_{k}$ is independent  and
every vertex of $Q_{k+1}$ is adjacent to vertices of at most two different color classes. 
Hence, we need at least $k-1$ vertices from $Q_{k+1}$
that make the connections between the vertices of $D_i$ colored with $\{1,\ldots,k\}$. 
The above statement along with the assumption $\vert D_i\vert=2k-1$ imply the claim. 
\end{proof}

\begin{claim}\label{cl:copies}
Let $i\in \{2,\ldots\ell-1\}$. Let $v\in D_i$ and $w\in D_{i+1}$ such that $v,w\notin V(H^0)\cup V(H^{k+1})$, at most 
one vertex in $\{v,w\}$ is in $V(H^{(1,1)})$, and $\widehat{c}(v)=\widehat{c}(w)\in \{1,\ldots,k\}$. 
Then, $v$ and $w$ are copies of the same vertex in $G$. Moreover, 
$v,w\in V(H^j)\cup V(H^{j+1})$ for some $j\in [k-1]$. 
\end{claim}

\begin{proof}\renewcommand{\qedsymbol}{\ensuremath{\lrcorner}}
Suppose $v$ and $w$ are not copies of the same vertex $z\in V(G)$. 
We know that $\vert D_i\vert=2k-1$ or $\vert D_i\vert=2k$. 

\medskip
\noindent
{\bf Case 1:  $\vert D_i\vert=2k-1$.} 
Since $D_i$ is a solution, 
$D_i$ induces a connected subgraph in $H$. 
By Claim~\ref{clm:allcolorsol}, $\vert D_i\cap Q_j\vert=1$ for all $j\in \{1,\ldots,k\}$ and $\vert D_i\cap Q_{k+1}\vert=k-1$. Also,  by  Claim~\ref{clm:Qindp}, $(i)$ $Q_1\cup \ldots \cup Q_k$ is an independent set in $H$, and  $(ii)$ every  vertex in $Q_{k+1}$ is connected to vertices of at most two distinct colors.
Statements $(i)$ and $(ii)$, and the fact that $\vert D_i\vert=2k-1$ imply that $(iii)$ $H[D_i]$ is a tree and each vertex in $D_i\cap Q_{k+1}$ is incident to exactly two vertices in $D_i$.
Since $\vert D_{i+1}\vert = \vert D_i\vert+1$, in reconfiguration step $i+1$, we add a vertex to obtain $D_{i+1}$. 
We know that $v\in D_i$. Since, for any color $q\in [k]$, there is exactly one vertex in $D_i$ of color $q$ 
(i.e., $\vert D_i\cap Q_q\vert=1$), we have that $D_{i+1}=D_{i} \cup \{w\}$. Moreover, in step $i+2$, the vertex removed 
from $D_{i+1}$ will be  from $\{v,w\}$ and that vertex will be $v$ (because of the minimality assumption of the length of the reconfiguration sequence). 
That is, $D_{i+2}=(D_{i}\cup \{w\})\setminus \{v\}$. 
Notice that $\vert D_i\vert=\vert D_{i+2}\vert=2k-1$. Let $b$ a vertex in $D_{i+2}$ which is adjacent to $w$ in 
$H[D_{i+2}]$. Since $Q_{k+1}\cap D_i=Q_{k+1}\cap D_{i+2}$ and $\vert D_i\vert=\vert D_{i+2}\vert=2k-1$, by Claim~\ref{clm:Qindp}, the neighbors of $b$ in $H[D_i]$ and $H[D_{i+2}]$ are of the same color. This implies that 
$b$ is adjacent to $v$ in $H[D_i]$.   
Thus, we proved that $\{b,w\},\{b,v\}\in E(H)$. 
If $b \in V(H^0)$, then $v,w\in V(H^{(1,1)})$ which is a contradiction to the assumption. Otherwise, by Claim~\ref{clm:samecopyvert}, we conclude that $v$ and $w$ are copies of same vertex.

\medskip
\noindent
{\bf Case 2:  $\vert D_i\vert=2k$. }
In this case $D_{i+1}$ is obtained by removing a vertex from $D_{i}$. Moreover, $i\geq 3$, because we have two vertices in $D_i$ from $V(H)\setminus D_1$. 
Since $\vert D_{i+1}\vert =2k-1$, because of Claim~\ref{clm:allcolorsol}, $D_{i+1}$ is obtained by removing the vertex $v$ from $D_i$. That is, $D_{i+1}=D_i\setminus \{v\}$ and $v,w\in D_{i}$. 
Then, again by Claim~\ref{clm:allcolorsol}, 
there is $v'\in \{v,w\}$ such that $D_{i-1}\uplus \{v'\}=D_{i}$. Let $w'=\{v,w\}\setminus \{v'\}$. Since $i\geq 3$, 
we now apply Case~1 with respect 
to $w'\in D_{i-1}$ and $v'\in D_i$ to complete the proof. 
%
\end{proof}



\begin{claim}
\label{claim:qinall}
For any index $j\in \{1,\ldots,k\}$ and color $q\in \{1,\ldots,k\}$, there exist an odd $i\in \{3,\ldots,\ell\}$ and $r\in \{5k,\ldots,15k\}$ such that $D_i$ contains a vertex of color $q$ from $V(H^{j,r})$. 
\end{claim}

\begin{proof}\renewcommand{\qedsymbol}{\ensuremath{\lrcorner}}
Without loss of generality, assume that $k\geq 2$. 
Moreover, for any odd $i\in [\ell-2]$, there is a vertex common in $D_i$ 
and $D_{i+2}$ (since $k\geq 2$). This implies that $H[D_1\cup D_3\ldots D_{\ell}]$ is a connected subgraph of $H$. Notice that for any $j\in \{1,\ldots, k\}$ and $r\in [20k]$, $V(H^{(j,r)})$ is a $(v_1,x_1)$-separator in $H$. 
Therefore, since $H[D_1\cup D_3\ldots D_{\ell}]$ is connected and $v_1,x_1\in D_1\cup D_{\ell}$, 
$(i)$ for any $j\in [k]$ and $r\in [20k]$,  there is an odd $i\in [\ell]$ such that $D_i$ contains a vertex from $V(H^{(j,r)})$. Now fix an index $j\in \{1,\ldots,k\}$ and a color $q\in \{1,\ldots,k\}$. By statement $(i)$, there is an odd $i\in \{1,\ldots,\ell\}$ such that $D_i$ contains a vertex from $V(H^{(j,10k)})$. 
Since $H[D_i]$ is connected, $\vert D_i\vert =2k-1$, $D_i\cap V(H^{(j,10k)})\neq \emptyset$,
and any vertex in $V(H)\setminus \bigcup_{r=5k}^{15k}V(H^{(j,r)})$ is at distance more that $5k$ (by the construction of $H$), we have that  all the vertices in $D_i$ belong to $\bigcup_{r=5k}^{15k}V(H^{(j,r)})$. 
Moreover, by Claim~\ref{clm:allcolorsol}, $D_i$ contains a vertex colored $q$ and that will also be present in $\bigcup_{r=5k}^{15k}V(H^{(j,r)})$.  This completes the proof of the claim. 
\end{proof}

\begin{claim}
\label{clm:allequal}
For any color $q\in \{1,\ldots,k\}$, the vertices of color $q$ from $\bigcup_{i=2}^k V(H^i)$ used in the reconfiguration sequence $D_1,\ldots,D_{\ell}$ are copies of the same vertex $z\in V(G)$.  
Moreover, exactly one vertex from $V(H^j)$ of color $q$ is used in the reconfiguration for all $2\leq j\leq k$.   
\end{claim}

\begin{proof}\renewcommand{\qedsymbol}{\ensuremath{\lrcorner}}
Fix a color $q\in \{1,\ldots,k\}$. 
By Claim~\ref{claim:qinall}, there are vertices of color $q$ from $V(H^j)$ for all $j$ is used in the reconfiguration sequence. By  Claim~\ref{cl:copies}, all these vertices are copies of the same vertex $z\in V(G)$. 
\end{proof}

Now we define a $k$-size vertex subset $C\subseteq V(G)$ and prove that $C$ is a clique in $G$. 
We let $C=\{a_i\in V(G) \colon 1\leq i\leq k, c(a_i)=i$, and the copy of $a_i$ in $V(H^2)$ is used in $D_1,\ldots,D_{\ell}\}$. 
Because of Claim~\ref{clm:allequal}, we have that $\vert C\vert =k$ and $C$ contains a vertex of each color in $c$. 
$C=\{a_1,\ldots,a_k\}\subseteq V(G)$ and for each $q\in [k]$, $c(a_q)=q$.  
We now prove that $C$ is indeed a clique in $G$. Towards that, we need to prove that for each $1\leq q<j \leq k$, 
$\{a_q,a_j\}\in E(G)$.

\begin{claim}
Let $1\leq q<j\leq k$. Then, $\{a_q, a_j\}\in E(G)$. 
\end{claim}

\begin{proof}\renewcommand{\qedsymbol}{\ensuremath{\lrcorner}}
By Claim~\ref{claim:qinall}, we know that there exist an odd $i\in [\ell]$ and $r\in \{5k,\ldots,15k\}$ such that $D_i$ contains a vertex of color $q$ in $V(H^{(j,r)})$. Thus, by Claim~\ref{clm:allequal}, a copy of $a_j$ and a copy of $a_q$  are present in $D_i$. 
%
%
Let $u_j$ and $u_q$ be the vertices in $D_i$ colored with $j$ and $q$, respectively. By Claim~\ref{clm:allequal}, $u_j$ is a copy of $a_j$ and $u_q$ is a copy of $a_q$. 
Any vertex $b$ in $V(H^j)$ colored $k+1$ is adjacent to 
vertices of exactly two colors, out of which one color is $j$.  Moreover, by the construction of $H$, $(a)$ if $b$ is adjacent to $x$ and $y$ in $V(H^j)$, and $x$ and $y$ are copies of $x'$ and $y'$ in $G$, respectively, then $\{x',y'\}\in E(G)$. We know that $H[D_i]$ is connected, $\vert Q_s\cap D_i\vert=1$ for all $1\leq s\leq k$, 
$D_i\setminus Q_{k+1}$ is an independent set in $H$, and each  vertex in $D_i$ colored with $k+1$ is adjacent to exactly 
two vertices in $D_i\setminus Q_{k+1}$  with one of them being $u_j$ 
(see Claims~\ref{clm:Qindp} and \ref{clm:allcolorsol}). This implies that there is common neighbor $b$ for $u_q$ and $u_j$ 
and hence $\{a_q,a_j\}\in E(G)$, by statement $(a)$ above. This completes the proof of the claim.  
\end{proof}
This completes the proof of the lemma. 
\end{proof}

\section{Details omitted from Section 4}
\subsection*{Proof of Lemma~\ref{lem:reduction}}

\begin{proof}
Let $H$ be the subgraph of $G$ induced by $D(u,v)$. 
We enumerate the vertices of 
$N(u)\cap N(v)$ consecutively as $x_1,\ldots, x_t$ for
some $t> 4|C| + 3k + 1$. We let $X = \{x_1,\ldots, x_t\}$. 
Note that since we have $t$ vertex-disjoint paths between $u$ 
and $v$ in $H$, these paths define the boundaries of $t$ faces in the plane 
embedding of $H$ (after applying the reduction rule of 
Lemma~\ref{lem:reduce-diamond-edges}, $H$ has all 
the edges $\{u,x\}$ and $\{v,x\}$ for $x\in N(u)\cap N(v)$
and no other edges). Each vertex in $C\setminus\{u,v\}$ 
can be adjacent in $H$ to at most two vertices in $X$, say 
with $y$ and~$z$, and these two vertices $y$ and $z$ 
can touch at most $3$ faces of $H$. 

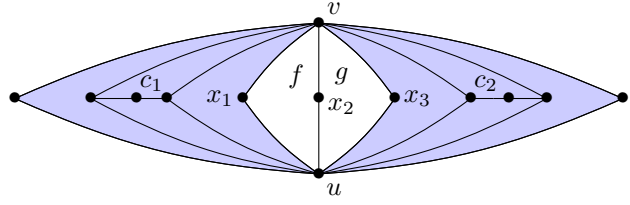
\begin{figure}[h!]
\begin{center}
\begin{tikzpicture}

\filldraw[fill=blue!20] (0,0) to[bend left=10] (-4,1) to[bend left=10] (0,2) to[bend right=10] (-1,1) to[bend right=10] cycle;
\filldraw[fill=blue!20] (0,0) to[bend right=10] (4,1) to[bend right=10] (0,2) to[bend left=10] (1,1) to[bend left=10] cycle;

\node (a1) at (0,0) {$\bullet$};
\node at (0.2,-0.2) {$u$};
\node (a2) at (0,2) {$\bullet$};
\node at (0.2,2.2) {$v$};

\node at (0,1) {$\bullet$};
\draw[-] (0,2) -- (0,1) -- (0,0);
\foreach \x in {-4cm,-3cm,-2cm,-1cm}
{
	\node at (\x,1) {$\bullet$};
	\draw[-] (0,2) to[bend right=10] (\x,1) to[bend right=10] (0,0);
}
\foreach \x in {1cm,2cm,3cm,4cm}
{
	\node at (\x,1) {$\bullet$};
	\draw[-] (0,2) to[bend left=10] (\x,1) to[bend left=10] (0,0);
}

\node at (-2.2,1.2) {$c_1$};
\node at (-2.4,1) {$\bullet$};
\draw (-2.4,1) -- (-3,1);
\draw (-2.4,1) -- (-2,1);

\node at (2.2,1.2) {$c_2$};
\node at (2.5,1) {$\bullet$};
\draw (2.3,1) -- (3,1);
\draw (2.3,1) -- (2,1);

\node at (-0.3,1.3) {$f$};
\node at (0.3,1.3) {$g$};

\node at (-1.3,1) {$x_1$};
\node at (1.3,1) {$x_3$};
\node at (0.3,0.9) {$x_2$};

\end{tikzpicture}
\end{center}
\caption{Every vertex of 
$C\setminus\{u,v\}$ can touch at most $3$ faces
of $H$. In the figure we assume the vertices
$c_1$ and $c_2$ are in $C\setminus\{u,v\}$. 
The faces that are touched by $c_1$ or $c_2$ 
are colored in blue. The uncolored 
faces $f$ and $g$ are not touched by vertices of 
$C\setminus\{u,v\}$. }
\end{figure}

This leaves $|C| + 3k + 1 > |C| + 1$ faces of $H$ that are not 
touched by a vertex of $C\setminus\{u,v\}$. 
By the pidgeonhole principle we can find $2$ adjacent 
faces~$f$ and $g$ of $H$ that are not touched by 
a vertex of $C\setminus \{u,v\}$. 

We let $x_1$ and $x_2$ denote the two vertices on the boundary of face $f$ different 
from $u$ and~$v$ and we let $x_2$ and~$x_3$ denote the two vertices on the boundary of face $g$ different from~$u$ 
and~$v$. 
Recall that, due to Lemma~\ref{lem:reduce-diamond-edges}, we know that there are no 
edges between those three vertices. Let $W$ denote the set
of all vertices contained in the face of the cycle $u,x_1,v,x_3,u$. In particular, $W$ contains $x_2$. We claim
that the vertices of~$W$ can be removed from $G$
without changing the reconfiguration properties of $G$, i.e., $W$ 
is a set of irrelevant vertices. Let $G'=G-W$. First observe that 
$W\cap (S\cup T)=\emptyset$, hence $S,T\subseteq V(G')$. 
We show that reconfiguration from $S$ to $T$ is possible
in $G$ if and only if reconfiguration from~$S$ to $T$ is 
possible in $G'$. 

Assume $S=S_1,\ldots, S_t=T$ is a reconfiguration sequence 
from $S$ to $T$ in $G$. 
Let $S_1', \ldots,S_t'$, where for $1\leq i\leq t$, $S_i'\coloneqq S_i$ if $S_i$ does not contain a vertex of $W$ 
and $S_i'\coloneqq (S_i\setminus W)\cup \{x_1\}$ otherwise. 
Note that this modification leaves $S$ and $T$ unchanged, 
hence, $S_1'=S_1$ and $S_t'=S_t$. 
We claim that $S_1', \ldots,S_t'$ is a reconfiguration sequence 
from $S$ to $T$ in $G'$. 

\setcounter{claimcounter}{0}
\begin{claim}
For $1\leq i\leq t$, $S_i'$ is a dominating set of $G$, and 
hence also of $G'$. 
\end{claim}

\begin{proof}\renewcommand{\qedsymbol}{$\lrcorner$}
No vertex of $W$ is adjacent to a vertex of 
$C\setminus\{u,v\}$ and $W \cap C = \emptyset$ by construction. Hence, the 
only vertices of $C$ that are possibly adjacent to a vertex
of $W$ are the vertices $u$ and $v$. Whenever $S_i$ contains a vertex of $W$, we have $x_1\in S_i'$, which dominates both $u$ and $v$. 
Hence, $S_i'$ dominates at least the vertices
of $C$ that $S_i$ dominates. We use 
Lemma~\ref{lem:non-adjecent-to-core-irrelevant} to
conclude that $S_i'$ is a dominating set of $G$. 
\end{proof}

\begin{claim}
For $1\leq i\leq t$, $S_i'$ is connected.
\end{claim}

\begin{proof}\renewcommand{\qedsymbol}{$\lrcorner$}
Let $s_1,s_2\in S_i\setminus W$
and let $P$ be a shortest path between $s_1$ and $s_2$ in $G[S_i]$. 
We have to show that there exists a path between $s_1$
and $s_2$ in $G[S_i']$. If $P$ does not use a vertex of $W$, then
there is nothing to show. Hence, assume $P$ uses a vertex of
$W$. By definition of $W$, both $s_1$
and $s_2$ lie outside the face $h$ of the cycle $u,x_1,v,x_3$
that contains $x_2$. Hence, $P$ must enter and leave the
face $h$, and as $P$ is a shortest path, it must enter and
leave via opposite vertices, i.e., via $u$ and $v$, or via 
$x_1$ and $x_3$ (as all other pairs are linked by an 
edge and we could find a shorter path). 
If $P$ contains $u$ and $v$, then 
we can replace the vertices of $W$ on~$P$ by $x_1$ and
we are done. 

Hence, assume $P$ uses $x_1$ and $x_3$. 
As $D(u,v)$ is a diamond of 
thickness greater than $4p + 3k + 1>3k$, according to 
Lemma~\ref{lem:diamond-ends-are-forced} at least 
one of the vertices $u$ and $v$, say $u$, is contained in~$S_i$,
and by definition also in~$S_i'$. Then 
we can replace the vertices of $W$ on~$P$ by~$u$ and we 
are again done. 
\end{proof}

Finally, the following claim is immediate from the definition of each $S_i'$. 
Combining Claim 1, 2, and 3, we conclude that $S_1',\ldots, S_t'$ is a reconfiguration sequence from $S$ to $T$ in $G'$. 

\begin{claim}
$S_{i+1}'$ is obtained from $S_i'$ by the addition or removal
of a single token for all $1\leq i<t$.  
\end{claim}

To prove the opposite direction, assume $S=S_1',\ldots, S_t'=T$ is a 
reconfiguration sequence from $S$ to $T$ in $G'$.
We claim that this is also a reconfiguration sequence from 
$S$ to $T$ in $G$. All we have to show is that 
$S_i'$ is a dominating set of $G$ for all $1\leq i\leq t$. 
This follows immediately from the fact that $S_i'$ is 
a dominating set of $G'$, and hence, as $W$ is not 
adjacent to $C\setminus\{u,v\}$ and $W \cap C = \emptyset$, also a dominating 
set of $C$ in $G$. Then according to 
Lemma~\ref{lem:non-adjecent-to-core-irrelevant}, 
$S_i'$ also dominates $G$. 
We conclude that there is a reconfiguration sequence from $S$ to $T$ in $G$ if and only if 
there is a reconfiguration sequence from $S$ to $T$ in $G'=G-W$. 
\end{proof}

\subsection*{Proof of Lemma~\ref{lem:reduction-paths}}

\begin{proof}
We first show that we can essentially establish the situation
depicted in Figure~\ref{fig:final-reduction}. 
We may assume that the paths of $\mathcal{P}$ are induced
paths, otherwise we may replace them by induced paths. 
Let $H$ be the graph induced on $u,v$ and the vertices 
of $\mathcal{P}$
that contains exactly the edges of the paths in $\mathcal{P}$. 
In the figure, the paths of $\Pp$ are depicted by thick
edges, while the diagonal edges do not belong to the paths. 
This situation is similar to the situation in the proof of
Lemma~\ref{lem:reduction}. Just as in the proof of 
Lemma~\ref{lem:reduction}, we find two adjacent faces
$f,g$ of $H$ that do not touch a vertex of $D\setminus \{u,v\}$. 

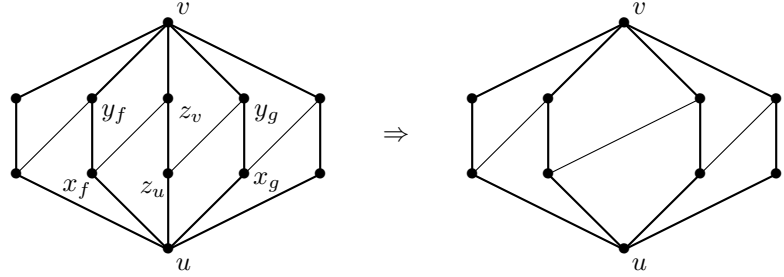
\begin{figure}[h!]
\begin{center}
\begin{tikzpicture}
\node (a1) at (0,0) {$\bullet$};
\node at (0.2,-0.2) {$u$};
\node (a2) at (0,3) {$\bullet$};
\node at (0.2,3.2) {$v$};
\foreach \x in {-2cm,-1cm,0cm,1cm,2cm}
{
	\node at (\x,1) {$\bullet$};
	\node at (\x,2) {$\bullet$};	
	\draw[-,thick] (0,3) -- (\x,2) -- (\x,1) -- (0,0);
}
\draw[-] (-2,1) -- (-1,2);
\draw[-] (-1,1) -- (0,2);
\draw[-] (0,1) -- (1,2);
\draw[-] (1,1) -- (2,2);

\node at (-1.2,0.8) {$x_f$};
\node at (-0.7,1.8) {$y_f$};
\node at (-0.2,0.8) {$z_u$};
\node at (0.3,1.8) {$z_v$};
\node at (1.3,0.9) {$x_g$};
\node at (1.3,1.8) {$y_g$};

\node at (3,1.5) {$\Rightarrow$};

\begin{scope}[xshift=6cm]
\node (a1) at (0,0) {$\bullet$};
\node at (0.2,-0.2) {$u$};
\node (a2) at (0,3) {$\bullet$};
\node at (0.2,3.2) {$v$};
\foreach \x in {-2cm,-1cm,1cm,2cm}
{
	\node at (\x,1) {$\bullet$};
	\node at (\x,2) {$\bullet$};	
	\draw[-,thick] (0,3) -- (\x,2) -- (\x,1) -- (0,0);
}
\draw[-] (-2,1) -- (-1,2);
\draw[-] (-1,1) -- (1,2);
\draw[-] (1,1) -- (2,2);
\end{scope}
\end{tikzpicture}
\end{center}
\caption{An exemplary situation handled by Lemma~\ref{lem:reduction-paths}.}
\label{fig:final-reduction}
\end{figure}

\setcounter{claimcounter}{0}
\begin{claim}
The paths bounding $f$ and $g$  have 
length $3$, i.e., they have exactly two inner vertices. 
\end{claim}

\begin{proof}\renewcommand{\qedsymbol}{$\lrcorner$}
First observe that $P\in \Pp$ cannot have length 
exactly $2$, as then $P$ contains a vertex
adjacent to both $u$ and $v$. However, the vertices 
with this property lie in $D$, and hence by construction 
not on $P$. 

Assume there is $P\in \Pp$ of length greater than $3$. 
Let $M(u)$ denote the neighbors of $u$ that 
are in $V(G) \setminus D$ and are only adjacent to $u$ 
and to no other vertex of $C$. 
Similarly, let $M(v)$ denote the neighbors of $v$ that are in $V(G) \setminus D$ and are 
only adjacent to $v$ and to no other vertex of $C$. 
By construction, the faces $f$ and $g$ do not
contain vertices of $D \setminus \{u,v\}$. Furthermore, $P$ contains
exactly one vertex of $M(u)$ and exactly one vertex of
$M(v)$. It cannot contain two vertices of one of these sets, 
as otherwise $P$ is not an induced path. Hence, assume that $P$
contains another vertex $x$ that is not in $M(u)\cup M(v)$. 
Then $x$ must be dominated by a vertex different from $u$
and from $v$. However, by construction, the faces $f$ and 
$g$ do not touch a vertex of $D \setminus \{u,v\} \supseteq (S \cup T) \setminus \{u,v\}$, a contradiction.  
\end{proof}

Denote by $x_f,y_f$ the two vertices that lie on the 
boundary of $f$ and not on the boundary of $g$ and 
by $x_g,y_g$ the two vertices that lie on the 
boundary of $g$ and not on the boundary of~$f$. 
Assume that $x_f,x_g\in M(u)$ and $y_f,y_g\in M(v)$. 
Denote by $z_u,z_v$ the vertices shared by~$f$ and~$g$
different from $u$ and $v$ that are adjacent to $u$ 
and $v$, respectively.  
Denote by $W$ the set of all vertices that lie inside the 
face~$h$ of the cycle $u,x_f,y_f,v,y_g,x_g,u$ that 
contains the vertices $z_u$ and $z_v$. Hence $W$ 
contains at least the vertices $z_u$ and $z_v$. 
By Corollary~\ref{lem:high-degree-left}, we know that $u,v\in S_i$, for all $1\leq i\leq t$ (both $u$ and $v$ can never be lifted). 
Consequently, by Lemma~\ref{lem:reduce-long-paths}, we know that there are no edges with both endpoints in $N(v)$ nor 
edges with both endpoints in $N(u)$. 
Combining the previous fact with the fact that all vertices of $W$ are adjacent
to either $u$ or $v$ (but not both) and to no other vertex of~$C \supseteq S \cup T$, 
we conclude that~$W$ 
consists of exactly the two vertices~$z_u$ and $z_v$ and that there are no edges between $z_u$ and $x_g,x_f$ 
and no edges between $z_v$ and $y_g,y_f$. Note that we can safely assume that none of the degree-one neighbors of 
$u$ or $v$ are inside $W$. 
We claim that the vertices 
$z_u$ and $z_v$ are irrelevant and can be removed after
possibly introducing an additional edge to the graph. 
Recall that $S$ and $T$ do not contain the vertices 
$z_u$ and $z_v$. We define $G'$ as follows. 

\begin{itemize}
\item If $\{u,v\} \not\in E(G)$ and ($\{x_f,z_v\}\in E(G)$ or $\{y_f,z_u\}\in E(G)$) and 
($\{x_g,z_v\}\in E(G)$ or $\{y_g,z_u\}\in E(G)$) then~$G'$ is obtained from $G$ by deleting $z_u$ and $z_v$ and introducing
the edge $\{x_f, y_g\}$. 
\item Otherwise, $G'$ is obtained from $G$ by simply deleting
$z_u$ and $z_v$. 
\end{itemize} 

We claim that $(G, S,T,k)$ and $(G',S,T,k)$ are equivalent
instances of \textsc{CDS-R}. 
Assume first that there exists a
reconfiguration sequence $S=S_1,\ldots, S_t=T$ in $G$. 
We distinguish two cases. First assume that $\{u,v\}\in E(G)$. Hence, $G'$ is obtained from $G$ by simply deleting
$z_u$ and $z_v$. 
Let $S_1', \ldots,S_t'$, where for $1\leq i\leq t$, 
$S_i'=S_i\setminus \{z_u,z_v\}$. 
We claim that $S_1', \ldots,S_t'$
is a reconfiguration sequence from $S$ to $T$ in $G'$. 

\begin{claim}
For $1\leq i\leq t$, $S_i'$ is a dominating set of $G$, and 
hence also of $G'$. 
\end{claim}

\begin{proof}\renewcommand{\qedsymbol}{$\lrcorner$}
The vertices $z_u$ and $z_v$ are not adjacent to a vertex of 
$C\setminus\{u,v\}$ and $\{z_u, z_v\} \cap C = \emptyset$. Hence, the 
only vertices of $C$ that are possibly adjacent to $z_u$ 
or $z_v$ are the vertices $u$ and $v$. According to 
Lemma~\ref{lem:high-degree-left}, 
$u,v\in S_i$, and moreover $u,v\in S_i'$, for all $1\leq i\leq t$. Hence, 
$S_i'$ dominates at least the vertices
of $C$ that $S_i$ dominates. We use 
Lemma~\ref{lem:non-adjecent-to-core-irrelevant} to
conclude that $S_i'$ is a dominating set of $G$. 
\end{proof}

\begin{claim}
For $1\leq i\leq t$, $S_i'$ is connected.
\end{claim}

\begin{proof}\renewcommand{\qedsymbol}{$\lrcorner$}
Let $s_1,s_2\in S_i\setminus \{z_u,z_v\}$
and let $P$ be a shortest path between $s_1$ and $s_2$ in $G[S_i]$. 
We have to show that there exists a path between $s_1$
and $s_2$ in $G[S_i']$. If $P$ does not use $z_u$ nor  
$z_v$ then there is nothing to prove. Hence, assume $P$ uses $z_u$ or 
$z_v$ (or both). By definition of~$W$, both $s_1$
and $s_2$ lie outside the face $h$ of the cycle $u,x_f,y_f,v,y_g,x_g,u$
that contains~$z_u,z_v$. Hence, $P$ must enter and leave the
face $h$, say it enters at $u$ and leaves at~$y_f$. All other
possibilities are handled analogously. Then we can avoid the 
vertices $z_u$ and $z_v$ by walking to $v$ first, then $u$ (or $x_f$), and then to $y_f$. 
\end{proof}

The next claim follows from the definition of $S_i'$ and the fact that we can remove any duplicate consecutive 
sets in a reconfiguration sequence.

\begin{claim}
$S_{i+1}'$ is 
obtained from~$S_i'$ by the addition or removal
of a single token for all $1\leq i<t$.  
\end{claim}

This finishes the proof in case $\{u,v\}\in E(G)$. Hence,  
we assume now that $\{u,v\}\not\in E(G)$ and ($\{x_f,z_v\}\in E(G)$ or $\{y_f,z_u\}\in E(G)$) and 
($\{x_g,z_v\}\in E(G)$ or $\{y_g,z_u\}\in E(G)$). That is, $G'$ is obtained from $G$ by deleting $z_u$ and $z_v$ and 
introducing the edge $\{x_f, y_g\}$. 
We now obtain $S_i'$ from $S_i$, for $1\leq i\leq t$,  
by replacing
\begin{itemize}
\item $z_u$ by $x_f$ and $z_v$ by $y_g$ if $S_i\cap \{z_u,z_v\}=\{z_u,z_v\}$, 
\item $z_u$ by $x_f$ if $S_i\cap \{z_u,z_v\}=\{z_u\}$, and 
\item $z_v$ by $y_g$ if $S_i\cap \{z_u,z_v\}=\{z_v\}$. 
\end{itemize}

We claim that $S_1', \ldots, S_t'$
is a reconfiguration sequence from $S$ to $T$ in $G'$. We need
no new arguments to prove that each $S_i'$ is a dominating set
of $G$ and hence of $G'$ and that each~$S_{i+1}'$ is obtained
from $S_i'$ by adding or removing one token. It remains
to show that each $S_i'$ is connected in $G'$. 

\begin{claim}
For $1\leq i\leq t$, $S_i'$ is connected in $G'$. 
\end{claim}

\begin{proof}\renewcommand{\qedsymbol}{$\lrcorner$}
According to 
Lemma~\ref{lem:high-degree-left}, 
$u,v\in S_i$, and also $u,v\in S_i'$, for all $1\leq i\leq t$.
If $S_i\setminus \{z_u,z_v\}$ is connected, $S_i'$ is also connected, 
hence assume $S_i\setminus \{z_u,z_v\}$ is not connected. 
As $X=\{u, x_f, z_u, x_g\}$ is connected via $u$
and $Y=\{v, y_f, z_v,y_g\}$ is connected via $v$, it suffices to 
show that our vertex exchange creates a connection in $G'$ between 
any vertex of $X$ and any vertex of $Y$. If 
$S_i\cap \{z_u,z_v\}=\{z_u,z_v\}$
this is clear, as we shift the tokens to $x_f$ and $y_g$ and
in $G'$ we have introduced the edge $\{x_f,y_g\}$. 
If $S_i\cap \{z_u,z_v\}=\{z_u\}$, then $\{z_u,y_g\}\in E(G)$ and $y_g\in S_i$, 
or $\{z_u,y_f\}\in E(G)$ and $y_f\in S_i$. We move the 
token $z_u$ to $x_f$. In the first
case we have connectivity via the new edge 
$\{x_f,y_g\}\in E(G')$, and in the second case we have 
connectivity via the edge $\{x_f,y_f\}\in E(G)$. The case 
$S_i\cap \{z_u,z_v\}=\{z_v\}$ is symmetric. 
\end{proof}


This finishes the proof that if $(G,S,T,k)$ is a positive instance  
then $(G',S,T,k)$ is a positive instance. 
Now assume that there
exists a reconfiguration sequence $S=S_1',\ldots, S_t'=T$ in $G'$.
In case we do not introduce the new edge to obtain $G'$ from 
$G$, we do not need new arguments to see that $S_1',\ldots, S_t'$
is a reconfiguration sequence also in $G$. 
Moreover, if $G''[S_i']$ is connected for all $i$, where $G''$ is 
obtained from $G'$ by removing the edge $\{x_f, y_g\}$, then again there is nothing to prove 
as $G'$ is a subgraph of $G$ and therefore $S=S_1',\ldots, S_t'=T$ is a reconfiguration sequence in $G$. 
Hence, assume that there exists at least one contiguous subsequence $\sigma$ starting at index $s$ and ending at index $f$ (with possibly $s = f$)
such that $G''[S_s'], G''[S_{s+1}'], \ldots, G''[S_f']$ are not connected. In other words, there exists a subsequence 
of length one or more that uses the edge $\{x_f, y_g\}$ for connectivity. 
Moreover, we assume, without loss of generality (the other case is symmetric), that $S_s'$ is obtained from $S_{s-1}'$ by adding a token on vertex $y_g$, i.e., 
$S_s' = S_{s-1}' \cup \{y_g\}$, and $S_{f+1}'$ is obtained from $S_f'$ by removing the token on vertex $x_f$, i.e., $S_{f+1}' = S_{f}' \setminus \{x_f\}$. 
We also assume that $E(G)$ contains the edges $\{x_f,z_v\}$ and $\{z_u,y_g\}$ (the remaining cases are handled identically). 
It remains to show how to modify $\sigma$ so that it does not use the edge $\{x_f, y_g\}$ for connectivity and remains a valid reconfiguration sequence in $G$. 
By applying the same arguments for any such subsequence we obtain the required reconfiguration sequence in $G$. 
We modify $\sigma$ as follows. We let $S_i'' = (S_i' \setminus \{y_g\}) \cup \{z_v\}$, for $s \leq i \leq f$. 
Then we replace $S_{f+1}'$ by four new sets $A_1$, $A_2$, $A_3$, and $A_4$, where $A_1 = S_{f}' \setminus \{x_f\}$, 
$A_2 = A_1 \cup \{z_u\}$, $A_3 = A_2 \setminus \{z_v\}$, $A_3 = A_3 \cup \{y_g\}$, and $A_4 = A_3 \setminus \{z_u\}$. 
Using the fact that the vertices $x_f,y_f,x_g,y_g$ are not adjacent to vertices of $D \setminus \{u,v\}$, 
it is easy to see that this yields a valid reconfiguration sequence, as both domination and connectivity are preserved. 
This completes the proof of the lemma. 
\end{proof}

\end{document}